\documentclass[a4paper,12pt]{amsart}

\usepackage[lmargin=2.4cm,rmargin=2.25cm,bmargin=3cm,tmargin=2.5cm]{geometry}

\usepackage[utf8]{inputenc}
\usepackage{amsmath}
\usepackage{amssymb}
\usepackage{amsthm}
\usepackage{textcomp}
\usepackage{enumerate}
\usepackage[round,sort]{natbib}
\usepackage{array}
\usepackage{bbm}
\usepackage{verbatim}
\usepackage{hyperref}
\usepackage{graphicx}
\usepackage{booktabs}
\usepackage{algorithm}
\usepackage{algpseudocode}
\usepackage{multicol}
\usepackage{cancel}  
\usepackage{subcaption}

\usepackage{xcolor}
\hypersetup{
  colorlinks,
  linkcolor={red!50!black},
  citecolor={blue!50!black},
  urlcolor={blue!80!black}
}

\newcommand{\defeq}{\mathrel{\mathop:}=}
\newcommand{\eqdef}{=\mathrel{\mathop:}}

\newcommand{\G}{\mathcal{G}}
\newcommand{\X}{\mathsf{X}}  
\newcommand{\Sp}{\mathsf{S}} 

\newcommand{\uarg}{\,\cdot\,}
\newcommand{\ud}{\mathrm{d}}
\newcommand{\R}{\mathbb{R}}
\newcommand{\N}{\mathbb{N}}
\renewcommand{\P}{\mathbb{P}}
\newcommand{\E}{\mathbb{E}}
\newcommand{\charfun}[1]{\mathbbm{I}\left\{#1\right\}}
\newcommand{\given}{\,:\,}
\newcommand{\var}{\mathrm{var}}

\newcommand{\tv}{\mathrm{tv}}
\newcommand{\lest}{\le_{\mathrm{st}}}
\newcommand{\gest}{\ge_{\mathrm{st}}}
\newcommand{\bigmid}{\;\big|\;}
\newcommand{\eqd}{\overset{d}{=}}

\newcommand{\theepsilon}{\boldsymbol{\epsilon}}
\newcommand{\thedelta}{\boldsymbol{\delta}}

\newcommand{\Categorical}{\mathrm{Categ}}

\newtheorem{theorem}{Theorem}

\newtheorem{lemma}[theorem]{Lemma}

\theoremstyle{definition}
\newtheorem{definition}[theorem]{Definition}
\newtheorem{assumption}[theorem]{Assumption}

\theoremstyle{remark}
\newtheorem{remark}[theorem]{Remark}

\title{Coupled conditional backward sampling particle filter}
\author{Anthony Lee \and Sumeetpal S.~Singh \and Matti Vihola}
\address{Anthony Lee, School of Mathematics, University of Bristol,
  University Walk, Bristol BS8 1TW, United Kingdom}
\address{Sumeetpal S.~Singh,
  Department of Engineering, University of Cambridge, Trumpington
  Street, Cambridge CB2 1PZ, United Kingdom}
\address{Matti Vihola, Department of Mathematics and Statistics, University of
Jyväskylä P.O.Box 35, FI-40014 University of Jyväskylä, Finland}

\begin{document}

\maketitle

\begin{abstract} 
The conditional particle filter (CPF) is a promising algorithm for general
hidden Markov model smoothing. Empirical evidence suggests that the
variant of CPF with backward sampling (CBPF) performs well even with
long time series. Previous theoretical results have not been able to
demonstrate the improvement brought by backward sampling, whereas we
provide rates showing that CBPF can remain effective with a fixed
number of particles independent of the time horizon. Our result is
based on analysis of a new coupling of two CBPFs, the coupled
conditional backward sampling particle filter (CCBPF). We
show that CCBPF has good stability properties in the sense that with
fixed number of particles, the coupling time in terms of iterations
increases only linearly with respect to the time horizon under a
general (strong mixing) condition. The CCBPF is useful not only as a
theoretical tool, but also as a practical method that allows for
unbiased estimation of smoothing expectations, following the
recent developments by~\citet{jacob-lindsten-schon}. 
Unbiased estimation has many advantages, such as enabling the construction of
asymptotically exact confidence intervals and straightforward
parallelisation.
\end{abstract} 

\section{Introduction} 

The conditional particle filter (CPF) introduced by
\citet{andrieu-doucet-holenstein} is a Markov Chain Monte Carlo method
that produces asymptotically unbiased samples from the posterior
distribution of the states of a hidden Markov model. The CPF can be
made significantly more efficient by the inclusion of backward
sampling \citep{whiteley-backwards-note}, or equivalently ancestor
sampling \citep{lindsten-jordan-schon}, steps: we refer to the resulting algorithm
as the conditional backward sampling particle filter (CBPF). While there are
many empirical studies reporting on the effectiveness of the CBPF for
Bayesian inference and on its superiority over the CPF \citep[see,
e.g.,][Section~7.2.2]{fearnhead2018particle}, quantitative theoretical
guarantees for the CBPF are still missing. In contrast, the
theoretical properties of the CPF are much better understood
\citep{chopin-singh, lindsten-douc-moulines, andrieu-lee-vihola}.

\cite{chopin-singh} introduced a coupling construction, called the
coupled CPF (CCPF), to prove the uniform ergodicity of the CPF.
Recently,~\cite{jacob-lindsten-schon} identified the potential use of
the CCPF to produce unbiased estimators by exploiting a de-biasing
technique due to~\cite{glynn-rhee} \citep[see
also][]{jacob2017unbiased}. This is an important algorithmic
advancement to particle filtering methodology since unbiased
estimation is useful for estimating confidence intervals, allows
straightforward parallelisation, and when used within a stochastic
approximation context, such as the stochastic approximation expectation
maximisation (SAEM) scheme \citep{delyon-lavielle-moulines}, unbiased estimators ensure martingale
noise, which has good supporting theory.

The methodological contribution of this paper is a relatively simple yet
important algorithmic modification to the CCPF by extending the CCPF
to include backward sampling steps through an index-coupled version of
\citeauthor{whiteley-backwards-note}'s
(\citeyear{whiteley-backwards-note}) backward sampling CPF.  This
approach, which we call the coupled conditional backward sampling
particle filter (CCBPF), gives theoretical insight to the behaviour of
the CBPF. It can also be used practically to facilitate
unbiased estimation. The CCBPF appears to be far more stable than the
CCPF \citep{chopin-singh} (and \citet{jacob-lindsten-schon}'s variant
that uses coupled ancestor sampling within the CCPF). Under a general
(but strong) mixing condition, we prove (Theorem~\ref{thm:ccbpf-coupling-time})
that the coupling time of CCBPF grows
at most linearly with length of the data record when a fixed number of
particles are used, provided this fixed number is sufficiently large.
From a computational perspective, this makes the CCBPF algorithm more appropriate
than alternative coupled CPFs when the length of the data record is very large,
as one only needs to have memory linear in the length of the data record.

As an important corollary of our analysis of the CCBPF, we obtain new
convergence guarantees for the CBPF (Theorem~\ref{thm:main-cbpf}) that
verifies its superiority over the CPF. More specifically, this result
differs from existing time-uniform guarantees for the CPF
\citep{andrieu-lee-vihola,lindsten-douc-moulines} which require
(super)linear growth of the number of particles. Our result confirms
the long held view, stemming from numerous empirical studies, that the
CBPF remains an effective sampler with a fixed number of particles
even as the data record length increases. An important consequence of
a fixed number of particles is that the space complexity (the amount of memory
required) of the
algorithm is linear, as opposed to quadratic, in the length of the
data record, making it feasible to run on long data records without
exhausting the memory available on a computer. We remark that a variant of the
the CPF that is stable with a fixed number of particles is
the blocked version of the CPF introduced by~\cite{singh-lindsten-moulines},
but that algorithm and its analysis are substantially different.

We also complement the empirical findings of~\cite{jacob-lindsten-schon} by
showing quantitative bounds on the
`one-shot' coupling probability of CCPF, that is, probability of
coupling after a single iteration of the algorithm.
These results are noteworthy as CCPF is applicable in some scenarios where CCBPF
is not, for instance when the transition density is intractable.
With the minimal assumption of bounded potentials,
we prove (Theorem~\ref{thm:oneshot-bounds}) that the coupling
probability of CCPF is at least $1-O(N^{-1})$, similar to what is shown
for the CPF \citep{andrieu-lee-vihola,lindsten-douc-moulines}. However, the
constants involved grow very rapidly with $T$.
Under strong mixing conditions, we are able to give a more precise rate of
convergence as $T$ increases (Theorem~\ref{thm:oneshot-rate-result}), which
still requires an exponentially
increasing number of particles $N$ with $T$. (See \citet{andrieu-lee-vihola,lindsten-douc-moulines}
for the CPF's rate of convergence under similar strong mixing assumptions.)
Our rates for the CCPF may be very conservative as empirical evidence
\citep{jacob-lindsten-schon} suggests increasing $N$ linearly with $T$ may be
sufficient for some models, although there is also evidence that $N$ should grow
superlinearly with $T$ for other models.


\section{Notation and preliminaries} 

Throughout the paper, we assume a general state space $\X$, which is
typically $\R^d$ equipped with the Lebesgue measure.  However, our
results hold for any measure space $\X$ equipped with a $\sigma$-finite dominating
measure, which is denoted as `$\ud x$'. Product spaces
are equipped with the related product measures. We use the notation
$a{:}b=a,\ldots,b$ for any integers $a\le b$, and use similar
notation in indexing $x_{a:b} = (x_a,\ldots,x_b)$ and $x^{(a:b)} =
(x^{(a)},\ldots, x^{(b)})$. We also use combined indexing, such that
for instance $x_{1:T}^{(i_{1:T})} = (x_1^{(i_1)},\ldots,
x_T^{(i_T)})$. We adopt the usual conventions concerning empty products and sums, namely  $\prod_a^b(\uarg) = 1$ and
$\sum_a^b(\uarg) = 0$ when  $a>b$. We denote $x\wedge y \defeq
\min\{x,y\}$, $x\vee y \defeq \max\{x,y\}$ and $(x)_+\defeq x \vee 0$.

We use standard notation for the $k$-step transition probability of a
Markov kernel $P$ by $P^k(x,A) \defeq \int P(x,\ud y) P^{k-1}(y,A)$
and $P^0(x,A) \defeq \charfun{x\in A}$. If $\nu$ is a probability
measure and $f$ is a real-valued function, then $(\nu P)(A) \defeq
\int \nu(\ud x) P(x,A)$, $(P f)(x) = \int P(x,\ud y) f(y)$ and $\nu(f)
\defeq \int \nu(\ud x) f(x)$, whenever well-defined. The total
variation metric between two probability measures $\mu,\nu$ is defined
as $\| \mu - \nu \|_\tv \defeq \sup_{|f|\le 1} |\mu(f) - \nu(f)|$, and
$\| f\|_\infty \defeq \sup_x |f(x)|$. If two random variables $X$ and
$Y$ share a common law, we write $X\eqd Y$. We denote the
categorical distribution by $\Categorical(\omega^{(1:N)})$ for
unnormalised probabilities $\omega^{(i)}$, that is,
$I \sim \Categorical(\omega^{(1:N)})$ if
$\P(I=i) = \omega^{(i)}/\sum_{j=1}^N \omega^{(j)}$.

We are interested in computing expectations of smoothing functionals,
with respect to the probability density
$\pi_T(x_{1:T}) \defeq \gamma_T(x_{1:T})/c_T$
on a space $\X^T$ with the following unnormalised density
\citep[cf.][]{del-moral}:
\begin{equation}
    \gamma_T(x_{1:T}) \defeq
    M_1(x_1) G_1(x_1) \prod_{t=2}^T M_t(x_{t-1}, x_t)
    G_t(x_{t-1},x_t),
    \label{eq:pi}
\end{equation}
where
$M_1$ is a probability density, $M_t$ are Markov transition
densities, $G_1:\X\to[0,\infty)$ and
$G_t:\X^2\to[0,\infty)$ for $t\in\{2{:}T\}$ are `potential functions',
and $c_T\defeq \int \gamma_T(x_{1:T})\ud x_{1:T}\in (0,\infty)$ is an
unknown normalising constant. The probability density in~\eqref{eq:pi} also encompasses the posterior density of a
hidden Markov model (HMM) when the pair $(G_t,M_t)$ is defined appropriately. For example, let the HMM be defined by a
Markov state process having initial density $f_1(x_1)$ and
transition densities 
$f_{t}(x_t \mid x_{t-1})$, and an observed process having densities $g_{t}(y_t \mid x_t)$. Given the sequence of observations $y_{1:T}$,
the potentials can be taken to be of the form
\[
    G_1(x_1) = \frac{g_1(y_1\mid x_1) f_1(x_1)}{M_1(x_1)}\quad
    \text{and}\quad
    G_t(x_{t-1},x_t) = \frac{g_t(y_t\mid x_t) f_t(x_{t}\mid
      x_{t-1})}{M_t(x_{t-1},x_t)},
\]
in which case $\pi_T(x_{1:T})$ corresponds to the smoothing distribution
of the HMM, that is, the conditional density of the latent Markov
states given the observations.
In the fairly common case where $M_1=f_1$ and $M_t(x_{t-1},x_t)=f(x_t\mid x_{t-1})$
for $t > 1$, we write $G_t(x_t) := G_t(x_{t-1},x_t) = g_t(y_t \mid x_t)$ to
emphasise that $G_t$ is a function only of $x_t$.

We will consider {two} different conditions for the
model. Assumption~\ref{a:g-minimal} is generally regarded as non-restrictive in the particle filtering literature and
essentially equivalent with the
uniform ergodicity of CPF \citep{andrieu-lee-vihola}.

\begin{assumption}
    \label{a:g-minimal} (Bounded potentials) 
\\          There exists $G^*<\infty$ such that $G_t(\uarg)\le G^*$ for all
$t=1{:}T$.
\end{assumption} 

Assumption~\ref{a:mixing}, which subsumes Assumption~\ref{a:g-minimal}, is a much stronger assumption introduced to prove time-uniform error bounds
of particle filtering estimates \citep{del2001stability}.
It is typically verified for models where $\mathsf{X}$ is compact with
(i) transitions that are likely to move between any two regions of the space and
(ii) potentials that do not vary too much.
Theoretical results using Assumption~\ref{a:mixing} are often indicative of performance in models where it does not quite hold.
The assumption has been replaced by weaker but fairly involved assumptions by~\cite{whiteley2013stability} in the context of
time-uniform error bounds, who also briefly surveys the use of Assumption~\ref{a:mixing}.
Another brief survey can be found in~\cite{douc2011sequential}.
At present, many theoretical papers on particle filters use Assumption~\ref{a:mixing} as weaker alternatives
are very difficult to manipulate theoretically and also to verify.

\begin{assumption}
    \label{a:mixing} (Strong mixing) 
\\ $G_*(1) \defeq
      \inf_{x\in\X} G_1(x)>0$ and $G^*(1) \defeq
      \sup_{x\in\X} G_1(x)<\infty$, and for all $t=2{:}T$,
\begin{enumerate}[(i)]
    \item $M_*(t)
      \defeq \inf_{x,y\in\X} M_t(x,y) > 0$ and
      $M^*(t)\defeq \sup_{x,y\in\X} M_t(x,y)<\infty$,
    \item $G_*(t) \defeq
      \inf_{x,y\in\X} G_t(x,y)>0$ and $G^*(t) \defeq
      \sup_{x,y\in\X} G_t(x,y)<\infty$.
\end{enumerate}
Denote $\thedelta \defeq \min_{t=1:T}
\frac{G_*(t)}{G^*(t)}$ and
$\theepsilon \defeq
\min_{t=1:T-1}
\frac{G_*(t) M_*(t) G_*(t+1) }{G^*(t) M^*(t) G^*(t+1)}$.
\end{assumption} 

\begin{remark} 
Note that $\thedelta \ge \theepsilon > 0$. The expression of constant $\theepsilon$ may be simplified (and
improved) in two special cases, as follows:
\begin{enumerate}[(i)]
    \item If $M_t(x,y) =
      M_t(y)$ for all $t=2{:}T$, then $M_*(t)/M^*(t)$
      may be omitted.
    \item If $G_t(x,y) = G_t(y)$ for $t=2{:}T$, then $G_*(t+1)/G^*(t+1)$
      may be omitted.
\end{enumerate}
In particular, if both hold, then $\theepsilon=\thedelta$.
When results are stated in the asymptotic regime where $T \to \infty$,
Assumption~\ref{a:mixing} should be interpreted as holding for all $T > 1$
with a uniform lower bound $\theepsilon > 0$.

\end{remark} 


\section{Convergence of the conditional backward sampling particle filter}

Before going to the construction of the coupled conditional particle
filters, we consider the conditional particle filter (CPF)
\citep{andrieu-doucet-holenstein} and the conditional backward
sampling particle filter (CBPF) \citep{whiteley-backwards-note} (in
short C\textsc{x}PF where \textsc{x} is a place holder), given in
Algorithm~\ref{alg:cxpf}. Both CPF and CBPF define a reversible Markov
transition with respect to $\pi_T$ \citep{chopin-singh}, with any
choice of the parameter $N\ge 2$ (number of particles).
The important distinction is that in the CPF, $X^{(J_{1:T})}_{1:T}$ is obtained
by tracing the ancestral line of $X_T^{(J_T)}$, whereas in the CBPF ancestors
are selected randomly according to the Markov transition densities and potential
functions.
The CPF and the CBPF have the same time complexity and space complexity, $O(TN)$.
The CBPF takes a constant factor more time due to the additional computations required
to sample random ancestors: if the transition densities are not expensive, this factor will
typically be less than $2$.

We focus first on the important implication of our result for the convergence time of
the CBPF, which applies also to the ancestor sampling implementation
of~\cite{lindsten-jordan-schon} as it is probabilistically equivalent to the
CBPF.

\begin{algorithm}[t]
    \caption{C\textsc{x}PF($X_{1:T}^*,N$).}
    \label{alg:cxpf} 
    \begin{algorithmic}[1]
        \State $X_{1:T}^{(1)} \gets X_{1:T}^*$.
        \State $X_{1}^{(i)} \sim M_1(\uarg)$
        for $i\in \{2{:}N\}$.
        \State $\omega_1^{(i)} \gets G_1(X_1^{(i)})$
        for $i\in \{1{:}N\}$.
        \For{$t=2{:}T$}
        \State $I_t^{(i)} \sim
        \Categorical(\omega_{t-1}^{(1:N)})$ for $i\in \{2{:}N\}$.
        \State $X_t^{(i)} \sim
        M_t(X_{t-1}^{(I_t^{(i)})}, \uarg)$
        for $i\in \{2{:}N\}$.
        \State $\omega_t^{(i)} \gets G_t(X_{t-1}^{(I_t^{(i)})},
        X_t^{(i)})$  for $i\in \{1{:}N\}$.
        \EndFor
        \State $J_T \sim
        \Categorical\big(\omega_T^{(1:N)}\big)$
        \For{$t=(T-1){:}1$}
        \State \textbf{if} \text{CBPF} \textbf{do} \label{line:cxpf-start}
        \State \hskip\algorithmicindent
        $b_t^{(i)} \gets \omega_t^{(i)} M_{t+1}(X_t^{(i)},
        X_{t+1}^{(J_{t+1})}) G_{t+1}(X_t^{(i)}, X_{t+1}^{(J_{t+1})})$
        \State \hskip\algorithmicindent
        $J_t \sim
        \Categorical\big( b_t^{(1:N)}\big)$
        \State \textbf{if} \text{CPF} \textbf{do}
       \State \hskip\algorithmicindent$J_t \gets
       I_{t+1}^{(J_{t+1})}$
       where $I_t^{(1)} = 1$.
       \label{line:cxpf-stop}
        \EndFor
        \State \textbf{output}
        $X_{1:T}^{(J_{1:T})}$
    \end{algorithmic}
\end{algorithm}

\begin{theorem}
    \label{thm:main-cbpf} 
Suppose Assumption~\ref{a:mixing} (strong mixing) holds, and denote by
$P_{T,N}$ the Markov transition probability of CBPF with $N$ particles
(Algorithm~\ref{alg:cxpf}). For any $\rho > 0$, there exists
$N_0=N_0(\theepsilon,\rho)<\infty$
such that for all $N \geq N_0$,
\begin{equation}
\label{eq:linear_convergence}
\lim_{T\to\infty} \sup_{x\in\X} \| P_{T,N}^{\lceil\rho T\rceil}(x, \uarg)  - \pi_T \|_\tv = 0.
\end{equation}
More precisely, for any $\alpha,\beta\in(1,\infty)$
there exists $N_0=N_0(\theepsilon,\alpha,\beta)\in\N$ such that for all $N\ge N_0$:
\begin{enumerate}[(i)]
\item
  \label{item:cbpf-convergence-bound} $\sup_{x\in\X} \| P_{T,N}^k(x, \uarg)  - \pi_T \|_\tv
    \le \alpha^T \beta^{-k}$ for all $k\ge 1$
    and all $T\ge 1$.
\item \label{item:cbpf-convergence-time} For any $\rho
  > \log \alpha/\log \beta$,~\eqref{eq:linear_convergence} holds.
\end{enumerate}
\end{theorem} 
\begin{proof} 
The upper bound~\eqref{item:cbpf-convergence-bound}
follows from Theorem~\ref{thm:ccbpf-coupling-time} and Lemma~\ref{lem:marginal-corollary}, and~\eqref{item:cbpf-convergence-time}
follows directly from~\eqref{item:cbpf-convergence-bound}. The first statement follows because
$\log \alpha / \log \beta$ can be taken to be arbitrarily small.
\end{proof} 
Theorem~\ref{thm:main-cbpf}, indicates that under the strong mixing
assumption, the mixing time of CBPF increases at most linearly in the
number of observations $T$. We remark that unlike existing results for
the CPF, we do not derive a one-shot coupling bound
\citep{chopin-singh}, or a one-step minorisation measure
\citep{andrieu-lee-vihola, lindsten-douc-moulines}, to prove the
uniform ergodicity of the CBPF transition probability $P_{T,N}$.  This
is because the enhanced stability of CBPF's Markov kernel over the
Markov kernel of CPF can only be established by considering the
behaviour of the iterated kernel $P_{T,N}^k$ of Theorem~\ref{thm:main-cbpf}, which has thus far proven elusive to study. Thus,
in addition to the result, the proof technique is itself novel and of
interest. For this reason we dedicate Section~\ref{sec:ccbpf-proof} to
its exposition.

\begin{remark}
\label{rem:quant_alpha_beta_rho_N0}
Intuitively, the arguments used to prove Theorem~\ref{thm:main-cbpf} demonstrate
that by increasing $N_0$, one can take $\rho$ to be smaller.
Hence, the qualitative relationship that increasing the number of particles
gives faster convergence of the CBPF is captured.
However, if one was to pursue quantitative bounds on the dependence
between $\rho$ and $N_0$, our bounds are likely to be too conservative to be
useful.
\end{remark}


\section{Coupled conditional particle filters}
\label{sec:algorithm} 

This section is devoted to the CCPF and CCBPF algorithms
(in short CC\textsc{x}PF where \textsc{x} is a place holder).
We start with Algorithm~\ref{alg:ccxpf}, where the
CC\textsc{x}PF algorithms are given in pseudo-code. The algorithms
differ only in lines~\ref{line:ccxpf-start}--\ref{line:ccxpf-stop},
highlighting the small, but important, difference:
the CCBPF incorporates index coupled backward sampling, which is
central to our results.

Algorithm~\ref{alg:cres} details the index
coupled resampling \citep{chopin-singh}, implementing maximal coupling
of $\Categorical(\omega^{(1:N)})$ and
$\Categorical(\tilde{\omega}^{(1:N)})$.
Line~\ref{line:uncoupled-mutation} of Algorithm~\ref{alg:ccxpf}
accomodates any sampling strategy which satisfies
$X_t^{(i)}\sim M_t(X_{t-1}^{(I_t^{(i)})}, \uarg)$
and $\tilde{X}_t^{(i)}\sim M_t(\tilde{X}_{t-1}^{(\tilde{I}_t^{(i)})},
\uarg)$ marginally, but may involve dependence, such as
implementation using common
random number generators \citep{jacob-lindsten-schon}.

\begin{algorithm}[t]
    \caption{CC\textsc{x}PF($X_{1:T}^*,\tilde{X}_{1:T}^*,N$).}
    \label{alg:ccxpf} 
    \begin{algorithmic}[1]
        \State $\big(X_{1:T}^{(1)}, \tilde{X}_{1:T}^{(1)}\big)
        \gets \big(X_{1:T}^*, \tilde{X}_{1:T}^*\big)$.
        \State $X_{1}^{(i)}\gets \tilde{X}_1^{(i)} \sim M_1(\uarg)$
        for $i\in \{2{:}N\}$.
        \State $\omega_1^{(i)} \gets G_1(X_1^{(i)})$;
        $\tilde{\omega}_1^{(i)} \gets G_1(\tilde{X}_1^{(i)})$.
        \For{$t=2{:}T$}
        \State $(I_t^{(2:N)}, \tilde{I}_t^{(2:N)}) \gets
        \textsc{CRes}\big(\omega_{t-1}^{(1:N)},
        \tilde{\omega}_{t-1}^{(1:N)}, N-1\big)$.
        \State $X_t^{(i)} \gets \tilde{X}_t^{(i)} \sim
        M_t(X_{t-1}^{(I_t^{(i)})}, \uarg)$
        for $i\in \{2{:}N\}$ with
          $X_{t-1}^{(I_t^{(i)})} = \tilde{X}_{t-1}^{(\tilde{I}_t^{(i)})}$
        \State
        \label{line:uncoupled-mutation}
        $(X_t^{(i)},\tilde{X}_t^{(i)})  \sim \big(
        M_t(X_{t-1}^{(I_t^{(i)})}, \uarg),
        M_t(\tilde{X}_{t-1}^{(\tilde{I}_t^{(i)})},
        \uarg)\big)$
        for $i\in \{2{:}N\}$ with
          $X_{t-1}^{(I_t^{(i)})} \neq \tilde{X}_{t-1}^{(\tilde{I}_t^{(i)})}$
        \State $\omega_t^{(i)} \gets G_t(X_{t-1}^{(I_t^{(i)})},
        X_t^{(i)})$; $\tilde{\omega}_t^{(i)} \gets
        G_t(\tilde{X}_{t-1}^{(\tilde{I}_t^{(i)})},
        \tilde{X}_t^{(i)})$.
        \EndFor
        \State $(J_T,\tilde{J}_T) \gets
        \textsc{CRes}\big(\omega_T^{(1:N)},
                     \tilde{\omega}_T^{(1:N)} , 1\big)$
        \For{$t=(T-1){:}1$}
        \State \textbf{if} \text{CCBPF} \textbf{do} \label{line:ccxpf-start}
        \State \hskip\algorithmicindent $b_t^{(i)} \gets \omega_t^{(i)} M_{t+1}(X_t^{(i)},
        X_{t+1}^{(J_{t+1})}) G_{t+1}(X_t^{(i)}, X_{t+1}^{(J_{t+1})})$
        \State \hskip\algorithmicindent $\tilde{b}_t^{(i)} \gets \tilde{\omega}_t^{(i)}
        M_{t+1}(\tilde{X}_t^{(i)},
        \tilde{X}_{t+1}^{(\tilde{J}_{t+1})}) G_{t+1}(\tilde{X}_t^{(i)},
        \tilde{X}_{t+1}^{(\tilde{J}_{t+1})})$
        \State \hskip\algorithmicindent $(J_t, \tilde{J}_t) \gets
        \textsc{CRes}\big( b_t^{(1:N)}, \tilde{b}_t^{(1:N)}, 1\big)$
        \State \textbf{if} \text{CCPF} \textbf{do}
       \State \hskip\algorithmicindent$(J_t,\tilde{J}_t) \gets
       (I_{t+1}^{(J_{t+1})}, \tilde{I}_{t+1}^{(\tilde{J}_{t+1})})$
       where $I_t^{(1)} =
        \tilde{I}_t^{(1)} = 1$.
       \label{line:ccxpf-stop}
        \EndFor
        \State \textbf{output}
        $(X_{1:T}^{(J_{1:T})},\tilde{X}_{1:T}^{(\tilde{J}_{1:T})})$
    \end{algorithmic}
\end{algorithm}

\begin{algorithm}[t]
    \caption{\textsc{CRes}$(\omega^{(1:N)},\tilde{\omega}^{(1:N)},n)$.}
    \label{alg:cres} 
    \begin{algorithmic}[1]
    \State $w^{(1:N)} \gets \frac{\omega^{(1:N)}}{\sum_{j=1}^N
      \omega^{(j)}}$;
    $\tilde{w}^{(1:N)} \gets \frac{\tilde{\omega}^{(1:N)}}{\sum_{j=1}^N
    \tilde{\omega}^{(j)}}$; $p_c \gets \sum_{j=1}^N
    w^{(j)}\wedge
      \tilde{w}^{(j)}$
    \State $w_c^{(1:N)} \gets \frac{w^{(1:N)}\wedge
      \tilde{w}^{(1:N)}}{p_c}$;
    $w_r^{(1:N)} \gets \frac{w^{(1:N)} - p_c w_c^{(1:N)}}{1-p_c}$;
    $\tilde{w}_r^{(1:N)} \gets \frac{\tilde{w}^{(1:N)} - p_cw_c^{(1:N)}}{1-p_c}$
    \For{$i=1{:}n$}
    \State \textbf{with probability} $p_c$ \textbf{do}
    \State \hskip\algorithmicindent $\tilde{I}^{(i)}\gets I^{(i)} \sim
    w_c^{(1:N)}$
    \State \textbf{otherwise}
    \State \hskip\algorithmicindent $I^{(i)} \sim
    w_r^{(1:N)}$; $\tilde{I}^{(i)} \sim \tilde{w}_r^{(1:N)}$.
    \EndFor
    \State \textbf{output} ($I^{(1:n)},\tilde{I}^{(1:n)}$)
    \end{algorithmic}
\end{algorithm} 

The CC\textsc{x}PF algorithms define Markov transition
probabilities on $\X^T\times\X^T$. The CC\textsc{x}PF algorithm is a Markovian
coupling of the corresponding C\textsc{x}PF algorithm, with the same structure: it is direct to check that
CC\textsc{x}PF coincides marginally with
C\textsc{x}PF in Algorithm~\ref{alg:cxpf}, that is,
if
$(S,\tilde{S}) \gets \text{CC\textsc{x}PF}(s_\mathrm{ref},
\tilde{s}_\mathrm{ref}, N)$ for some $N\ge 2$ and
$s_\mathrm{ref},\tilde{s}_\mathrm{ref}\in\X^T$, then $S \eqd
\text{C\textsc{x}PF}(s_\mathrm{ref},N)$ and $\tilde{S} \eqd
\text{C\textsc{x}PF}(\tilde{s}_\mathrm{ref},N)$. It is also clear that if
$s_\mathrm{ref}=\tilde{s}_\mathrm{ref}$, then $S=\tilde{S}$. Because
CPF and CBPF are both $\pi_T$-reversible \citep{chopin-singh}, it is
easy to see that CC\textsc{x}PF are $\boldsymbol{\pi}_T$-reversible,
where $\boldsymbol{\pi}_T(\ud s,\ud \tilde{s}) = \pi_T(s) \delta_s(\ud
\tilde{s}) \ud s$. Just as the CPF and CBPF algorithms have time and space complexity
$O(TN)$, so do the CCPF and CCBPF algorithms.

\subsection{Convergence of the CCBPF} 

In our experiments, the CCBPF had stable behaviour with a fixed
and small number of particles, even for large $T$. Our main result for the CCBPF
consolidates our empirical findings. In contrast
to most results for the CPF and CCPF, the statement of the coupling
behaviour for CCBPF is not one-shot in nature: instead we show
that the pair of trajectories output by the repeated application of
the CCBPF kernel couple themselves \emph{progressively}, starting from their
time $1$ components until eventually coupling all their components
until time $T$.

\begin{theorem}
    \label{thm:ccbpf-coupling-time} 
Suppose that Assumption~\ref{a:mixing} holds.
Let
$s_\mathrm{ref},\tilde{s}_\mathrm{ref}\in \X^T$
and let $(S_0,\tilde{S}_0) \gets (s_\mathrm{ref},\tilde{s}_\mathrm{ref})$
and $(S_k,\tilde{S}_k) \gets
\mathrm{CCBPF}(S_{k-1},\tilde{S}_{k-1},N)$ for $k\ge 1$.
Denote the coupling time $\tau \defeq \inf\{k\ge 1\given S_k = \tilde{S}_k\}$.
For any $\rho > 0$, there exists $N_0=N_0(\theepsilon,\rho)<\infty$
such that for all $N \geq N_0$,
\begin{equation}
\label{eq:linear-convergence-ccbpf}
\lim_{T\to\infty} \P(\tau \ge \lceil \rho T \rceil) = 0.
\end{equation}
More precisely, for any $\alpha,\beta\in(1,\infty)$
there exists
$N_0=N_0(\theepsilon,\alpha,\beta) \in\N$ such that for all $N\ge N_0$,
\begin{equation}
    \P(\tau \ge n ) \le \alpha^T \beta^{-n}, \qquad \text{for all
      $n,T\in\N$.}
    \label{eq:ccbpf-coupling-time}
\end{equation}
In particular, for any
$\rho > \log(\alpha)/\log(\beta)$,~\eqref{eq:linear-convergence-ccbpf} holds.
\end{theorem} 
The proof of the bound~\eqref{eq:ccbpf-coupling-time} is given in
Section~\ref{sec:ccbpf-proof}, and the linear coupling time statement
follows by appropriate choice of $\alpha$ and $\beta$.
The most striking element of this statement is that the coupling time $\tau$ does not exceed $\rho T$ with greater surety as $T$ increases.

\begin{remark}
The comments in Remark~\ref{rem:quant_alpha_beta_rho_N0} also apply here.
Regarding tightness of Theorem~\ref{thm:ccbpf-coupling-time}, one may consider
whether coupling might occur with high probability after a number of iterations that is sublinear in $T$.
The simulation experiments in Section~\ref{sec:empirical} suggest that the mean
coupling time for the CCBPF is indeed linear in $T$ for the models considered, suggesting that
Theorem~\ref{thm:ccbpf-coupling-time} is tight, and we suspect that this is the case for many other models.
Similarly, in the more challenging simulation experiment, we found that values
of $N$ that were too small had mean coupling times that were
superlinear in $T$,
suggesting that there is indeed a model-dependent, minimal $N_0$ for linear-in-time coupling.
\end{remark}


\subsection{Convergence of the CCPF} 

We seek here to provide quantitative results to strengthen Theorem 3.1
of~\cite{jacob-lindsten-schon}, which does not quantify the dependence of
the probability of coupling on $N$ or $T$.
Although the results are less encouraging than for the CCBPF, the dependence
on $T$ is likely to be very conservative for many models. On the other hand,
results for the CCPF are interesting because it is more widely applicable than
the CCBPF, specifically since implementing the CCPF does not require the
ability to calculate densities $M_t$.

In empirical investigations, we have only seen the CCPF couple instantaneously,
as opposed to the progressive coupling seen for the CCBPF.
For that reason, we focus here on lower bounding the one-shot
coupling probability for two arbitrary reference trajectories.

Our first result states that with $T$ fixed,
the CCPF enjoys similar strong uniform ergodicity to the CPF, with the same rate
as the number of particles N is increased \citep[cf.][]{andrieu-lee-vihola,lindsten-douc-moulines}.

\begin{theorem}
    \label{thm:oneshot-bounds} 
Let $s_\mathrm{ref},\tilde{s}_\mathrm{ref}\in \X^T$, and consider
$(S, \tilde{S})
\gets \mathrm{CCPF}(s_\mathrm{ref},\tilde{s}_\mathrm{ref},N)$ with $N\ge 2$.
If $G_t(x_{t-1},x_t)=G_t(x_t)$ for $t\ge 2$ and
Assumption~\ref{a:g-minimal} holds, then there exists a constant
$c=c(G^*, T, c_T)\in(0,\infty)$ such that
\[
    \P( S = \tilde{S}) \ge 1 - \frac{c}{N + c}.
\]
\end{theorem} 
The proof of Theorem~\ref{thm:oneshot-bounds} is given in
Appendix~\ref{app:oneshot}.

Theorem~\ref{thm:oneshot-bounds} is stated with a fixed time horizon
$T$, and shows that one-shot coupling occurs from any initial state
$(s_\mathrm{ref},\tilde{s}_\mathrm{ref})$ with positive probability for any
$N \geq 2$. To have a reasonably large probability of one-shot coupling, it is
sufficient to choose a large enough value of $N$.

Although Theorem~\ref{thm:oneshot-bounds} holds very generally, it
does not provide useful quantitative bounds
on the relationship between $N$ and $T$; in particular $c(G^*, T, c_T)$ may grow
very quickly with $T$ so that $N$ would need to grow similarly quickly to
control the coupling probability.
In order to provide more accurate bounds relating $N$ and $T$,
we have had to make the strong mixing assumption.
\begin{theorem}
    \label{thm:oneshot-rate-result} 
Under the setting of Theorem~\ref{thm:oneshot-bounds}, but
with Assumption~\ref{a:mixing},
\[
    \P( S = \tilde{S}) \ge 1 -
    \frac{2^T T}{(2c_*)^{-1}(N-1)+1}.
\]
\end{theorem}
Theorem~\ref{thm:oneshot-rate-result}, which follows from Lemmas~\ref{lem:mixing-to-mixing-type}
and~\ref{lem:oneshot-with-rate-proof} in Appendix~\ref{app:oneshot-rate},
shows that the probability of coupling does not diminish when $N =
O(2^T T)$. That is, roughly doubling of particle number with every
unit increase in $T$ ensures non-diminishing coupling probability.
Experiments by~\cite{jacob-lindsten-schon} and also those in Section~\ref{sec:empirical}
suggest that a rate $N = O(T)$ might be enough for some models, which would be
analogous to the results on the CPF \citep{andrieu-lee-vihola,lindsten-douc-moulines},
but we have been unable to verify such a rate theoretically.
Our empirical results in Section~\ref{sec:empirical} also suggest that taking $N$ superlinear in $T$ may be
necessary in other models.



\section{Unbiased estimators}
\label{sec:unbiased} 

Let us then turn to the use of CC\textsc{x}PF together with the scheme
of~\cite{glynn-rhee}, to construct unbiased estimators,
as suggested in \citep{jacob-lindsten-schon}.
\begin{algorithm}[t]
    \caption{\textsc{Unbiased}($h,b,N$)}
    \label{alg:unbiased} 
\begin{algorithmic}[1]
    \State \label{line:init1} Run particle filter with $(M_t,G_t)_{t=1:T}$
    independently to get trajectories $\tilde{S}_{0},S_{-1} \in \X^T$
    \State \label{line:init2} Set $(S_0,\,-\,) \gets
    \text{CC\textsc{x}PF}(S_{-1},S_{-1},N)$.
    \For{$n=1,2,\ldots$}
    \State  $(S_{n},\tilde{S}_{n}) \gets
    \text{CC\textsc{x}PF}(S_{n-1},\tilde{S}_{n-1},N)$
    \State \label{line:stopping}
    \textbf{if} $S_n = \tilde{S}_n$ and $n\ge b$ \textbf{then}
    \textbf{output} $Z \defeq h(S_b) + \sum_{k=b+1}^n [h(S_k) -
    h(\tilde{S}_k)]$
    \EndFor
\end{algorithmic}
\end{algorithm}
Algorithm~\ref{alg:unbiased} aims to unbiasedly estimate the
expectation $\E_{\pi_T}[h(X_{1:T})]$, where $h : \mathsf{X}^T \to
\mathbb{R}$ is any integrable function of interest. We remark that
the unbiased estimation procedure is not specific to the choice of
$h$.
For example, in an HMM one can approximate smoothing means and covariances
by choosing several appropriate $h$ functions.
Algorithm~\ref{alg:unbiased} has two adjustable parameters, a
`burn-in' $b\ge 1$ and `number of particles' $N\ge 2$ which may be
tuned to maximise its efficiency.  Algorithm~\ref{alg:unbiased}
iterates either the coupled conditional particle filter CCPF or the
coupled conditional backward sampling particle filter CCBPF until a
perfect coupling of the trajectories $S_n$ and $\tilde{S}_n$ is
obtained.

The following result records general conditions under which
the scheme above produces unbiased finite variance estimators.
\begin{theorem}
    \label{thm:minimal-consistency} 
Suppose $G_t(x_{t-1},x_t)=G_t(x_t)$ and Assumption~\ref{a:g-minimal} holds, $h:\X^T\to\R$ is
bounded and measurable. Then, Algorithm~\ref{alg:unbiased} with
CCPF, $b\ge 1$ and $N\ge 2$ satisfies,
denoting by $\tau$ the running time (iterations before producing
output):
\begin{enumerate}[(i)]
    \item \label{item:tau-finiteness}
      $\tau<\infty$ almost surely.
\item \label{item:norate-unbiased-finite-var}
$\E[Z] = \E_{\pi_T}[h(X)]$
and $\var(Z)<\infty$.
\item
  With the constant $c=c(G^*,T,c_T)\in(0,\infty)$ of
  Theorem~\ref{thm:oneshot-bounds},
  \begin{align*}
      \E[\tau ] &\le
      b + \big(\frac{c}{N+c}\big)^{b-1} \big(\frac{N+c}{N}\big),
      \\
      |\var(Z) - \var_{\pi_T}\big(h(X)\big)| &\le
  {
  16 \| \bar{h}\|_\infty^2 \big(\frac{N+c}{N}\big)^2
  \big(\frac{c}{N+c}\big)^{b/2}
   }
  \end{align*}
 where $\bar{h}(\cdot)=h(\cdot) - \pi_T(h).$
\end{enumerate}
\end{theorem} 
\begin{proof} 
Theorem~\ref{thm:oneshot-bounds}
implies that $\P(\tau > k) \le \big(\frac{c}{N+c}\big)^{k-1}$ for all
$k>b$, from which
\[
    \E[\tau] = \sum_{k\ge 0} \P(\tau > k)
    \le b + \sum_{k\ge b} \P(\tau > k)
    \le b + \big(\frac{c}{N+c}\big)^{b-1} \big(\frac{N+c}{N}\big)
    <\infty,
\]
and the bound on $|\var(Z) - \var_{\pi_T}\big(h(X)\big)|$ follows from
Lemma~\ref{lem:var_bound}. Part~\eqref{item:norate-unbiased-finite-var} follows
from Theorem~\ref{thm:coupling-unbiased} and Lemma~\ref{lem:cauchy-terms}.
\end{proof}

Theorem~\ref{thm:minimal-consistency} extends the consistency
result of~\cite{jacob-lindsten-schon}, by quantifying the convergence
rates. Fix $T$: if $N$ is large, then $\E[\tau]\approx b$, and if $b$ is large,
then $\var(Z)\approx \var_{\pi_T}\big(h(X)\big)$. As mentioned after
Theorem~\ref{thm:ccbpf-coupling-time} the growth of the constant $c$ with respect
to $T$ can be very rapid. In contrast, in case of the CCBPF, the results may be
refined as follows:

\begin{theorem}
    \label{thm:mixing-result} 
Suppose Assumption~\ref{a:mixing} holds,
let $\alpha,\beta\in(1,\infty)$ and let
$N_0\in\N$ be from Theorem~\ref{thm:ccbpf-coupling-time}.
Then, Algorithm~\ref{alg:unbiased} with CCBPF and $b\ge 1$ satisfies, with $\rho \defeq
\log\alpha/\log\beta$:
\begin{enumerate}[(i)]
    \item
      \label{item:mixing-tau}
      $\E[\tau] \le  b \vee \lceil \rho T \rceil + \alpha^T
          \beta^{-(\lceil \rho T \rceil \vee b) } (\beta -1)^{-1}$.
    \item
      \label{item:mixing-var}
      $\big|\var(Z) - \var_{\pi_T}\big(h(X)\big)\big|
         \le 16 \alpha^T (1-\beta^{-1})^{-2} \beta^{-b/2}
         \|\bar{h}\|_\infty^2$.\\
      In particular, if $b= \lceil \tilde{\rho} T \rceil $with any $\tilde{\rho}>2\rho$,
      $\big|\var(Z) - \var_{\pi_T}\big(h(X)\big)\big|\to 0$ as
      $T\to\infty$.
\end{enumerate}
\end{theorem} 
\begin{proof} 
The results follow from Theorem~\ref{thm:ccbpf-coupling-time}
and Lemma~\ref{lem:var_bound},
similarly as in the proof of Theorem~\ref{thm:minimal-consistency}.
\end{proof}

Note that the latter term in Theorem~\ref{thm:mixing-result}~\eqref{item:mixing-tau} is at most $(\beta -1)^{-1}$, showing that
the expected coupling time is linear in $T$. Theorem~\ref{thm:mixing-result}~\eqref{item:mixing-var}
may be interpreted so
that the CCBPF algorithm is almost equivalent with perfect sampling
from $\pi_T$, when $b$ is increased linearly with respect to $T$.

We conclude the section with a number of remarks about
Algorithm~\ref{alg:unbiased}:
\begin{enumerate}[(i)]
\item We follow~\cite{jacob-lindsten-schon} and suggest an initialisation
based on a standard particle filter in line~\ref{line:init1}.
However, this initialisation may be changed to any other scheme,
which ensures that $S_0$ and $\tilde{S}_1$ have identical
distributions. Our results above do not depend on the chosen
initialisation strategy.
\item The estimator $Z$ is constructed for a single function
  $h:\X^T\to\R$, but several estimators may be constructed
  simultaneously for a number of functions $h_1,\ldots,h_m$. In fact,
  as~\cite{glynn-rhee} note,
  if we let $\tau \defeq \inf \{n\ge b\given S_n=\tilde{S}_n\}$,
  we may regard the random signed measure
  \[
      \textstyle
      \hat{\mu}_b(\uarg) \defeq \delta_{S_b}(\uarg) + \sum_{k={b+1}}^\tau
      [\delta_{S_k}(\uarg) - \delta_{\tilde{S}_k}(\uarg)]
  \]
  as the output, which will satisfy the unbiasedness
  $\E[\hat{\mu}_b(\varphi)]
  = \pi_T(\varphi)$ at least for all bounded measurable $\varphi:\X\to\R$.
\item It is also possible to construct a `time-averaged' estimator that
  corresponds to a weighted average of the estimators $\hat{\mu}_b$ over a
  range of values for $b$ \citep{jacob2017unbiased}.
\item We believe that the method is valid also without
  Assumption~\ref{a:g-minimal}
but may exhibit poor performance --- similar to the conditional
particle filter, which is sub-geometrically ergodic with unbounded potentials
\citep{andrieu-lee-vihola}.
\end{enumerate}


\section{Coupling time of CCBPF}
\label{sec:ccbpf-proof} 

Consider now the Markov chain $(S_k,\tilde{S}_k)_{k\ge 1}$ defined by
Algorihm~\ref{alg:unbiased}, with the stopping criterion (line~\ref{line:stopping}) omitted.
Define the `perfect coupling boundary' as
\[
    \kappa_n \defeq \kappa(S_n, \tilde{S}_n)
    \defeq
    \max\big\{t\ge 0\given S_{n,1:t} =
    \tilde{S}_{n,1:t} \big\},
\]
We are interested in upper bounding the stopping time $\tau \defeq
\inf\{n\ge 1\given S_n = \tilde{S}_n\} = \inf\{n\ge 1\given \kappa_n =
  T \}$.

Since the CCBPF is complicated, in our analysis we instead focus on a simplified
Markov chain that considers only the vector of numbers of identical particles
at each time $t \in \{1{:}T\}$. The boundary associated with this simpler
chain grows by i.i.d. increments, which are stochastically ordered
with respect to the increments of the CCBPF boundary increments, ultimately
allowing us to upper bound the stopping time.

We use stochastic ordering $X\lest Y$ of
two random variables $X$ and $Y$, which holds
if their distribution functions are ordered
$\P(X\le x) \ge \P(Y\le x)$ for all $x\in\R$.
Two random vectors $X$ and $Y$ are ordered $X\lest Y$
if $\E[\phi(X)]\le \E[\phi(Y)]$ for all functions $\phi:\R^n\to\R$ for
which the expectations exist, and which are increasing, in the sense
that $\phi(x)\le \phi(y)$ whenever $x\le
y$, where `$\le$' is the usual partial order $x\le y$ if $x_i \le y_i$
for all $i\in\{1{:}d\}$.
Recall also that $X\lest Y$ if and only if
there exists a probability space with random
variables $\tilde{X}$ and $\tilde{Y})$  such that $X\eqd \tilde{X}$ and
$Y\eqd \tilde{Y}$ and $\tilde{X}\le \tilde{Y}$ a.s.
\cite[Theorem 6.B.1]{shaked-shanthikumar}.

Our bound of $\tau$ is based on
an independent random variable $\Delta$, which
satisfies $\kappa_{n+1} - \kappa_n \gest \Delta \wedge
(T-\kappa_n)$, under Assumption~\ref{a:mixing}.

\begin{lemma}
\label{lem:domination} 
Suppose Assumption~\ref{a:mixing} holds,
and consider the output of Algorithm~\ref{alg:ccxpf} (CCBPF).
The perfect coupling boundaries satisfy
\[
    \kappa(X_{1:T}^{(J_{1:T})}, \tilde{X}_{1:T}^{(\tilde{J}_{1:T})})
    - \kappa(X_{1:T}^*,\tilde{X}_{1:T}^*)
    \gest \Delta\wedge \big(T-\kappa(X_{1:T}^*,\tilde{X}_{1:T}^*)\big),
\]
where the random variable $\Delta$ is defined through the following
procedure:
\begin{enumerate}[(i)]
\item Let $\hat{C}_t = N$ for $t\le 0$ and $\hat{C}_1 = N-1$, and set $s=1$.
  While $\hat{C}_s>0$:
  \begin{itemize}
  \item Simulate
  $\hat{C}_{s+1} \sim \mathrm{Binom}\Big(N-1,  \frac{\thedelta \hat{C}_{s}}{
    \thedelta \hat{C}_{s} +
        (N-\hat{C}_{s})}\Big)$.
  \item Let $s \gets s+1$.
  \end{itemize}
\item Set $\hat{C}_t=0$ for $t>s$, $\xi_{s}=0$ and $t=s-1$. While $t\ge 0$ or $\xi_{t+1}=0$:
  \begin{itemize}
      \item Simulate $\xi_{t} \sim\mathrm{Bernoulli}(p_{t})$,
        where
        \[
            p_t \defeq \begin{cases}
            p_{t}^{(0)} \defeq
        \theepsilon \frac{\hat{C}_{t}}{N}, & \xi_{t+1}=0 \\
        p_{t}^{(1)} \defeq
        \frac{\hat{C}_{t}\theepsilon}{\hat{C}_{t}\theepsilon + N -
          \hat{C}_{t}}, & \xi_{t+1}=1.
        \end{cases}
        \]
      \item Let $t \gets t-1$.
  \end{itemize}
\item Set $\Delta \gets
  \min\{i\ge t\given  \xi_{i}=0\}-1$.
\end{enumerate}
\end{lemma}
\begin{proof} 
Denote in short $\kappa = \kappa(X_{1:T}^*,\tilde{X}_{1:T}^*)$,
and the indices of the coupled particles
\[
    C_t \defeq \{j \in \{1{:}N\} \given X_t^{(I_t^{(j)})} =
  \tilde{X}_t^{(\tilde{I}_t^{(j)})} \} \quad \text{for} \quad t\in\{1{:}T\}.
\]
Then, the sizes of $C_t$ satisfy the following:
\begin{align*}
    |C_t|&=N & \text{$t=1{:}\kappa$}, \\
    |C_{t}| \bigmid C_{1:t-1} &\gest \mathrm{Binom}\Big(N-1,
      \frac{\thedelta|C_{t-1}|}{\thedelta|C_{t-1}| +
        N-|C_{t-1}|}\Big) & \text{$t=(\kappa+1){:}T$,}
\end{align*}
where the latter follows by Lemma~\ref{lem:ic-res}~\eqref{item:ic-coupled}.
As the function $c\mapsto \thedelta c(N - (1-\thedelta)c)^{-1}$ is
increasing in $c$, and $\mathrm{Binom}(n,p) \gest
\mathrm{Binom}(n,p')$ for $p\ge p'$, it follows that
$(|C_{1}|,\ldots,|C_{T}|)\gest \hat{C}_{(1-\kappa):(T-\kappa)}$
\cite[Theorem 6.B.3]{shaked-shanthikumar}. This means that
we may construct (by a suitable coupling) $\hat{C}_t$ such that
$|C_t|\ge \hat{C}_{t-\kappa}$ for all $t\in\{1{:}T\}$.

By Lemma~\ref{lem:ic-res}, the backward sampling indices satisfy:
\begin{align*}
    \P(J_T = \tilde{J}_T\in C_T
    \mid C_{1:T}) &\ge  \frac{\thedelta|C_T|}{\thedelta |C_T| + N-|C_T|},
    \\
    \P(J_t=\tilde{J}_t\in C_t\mid C_{1:T}, J_{t+1:T})
    &\ge \begin{cases}
    \frac{\theepsilon|C_t|}{\theepsilon|C_t| + N-|C_t|},&\text{$J_{t+1}=\tilde{J}_{t+1}\in
      C_{t+1}$,}  \\
    \theepsilon \frac{|C_t|}{N},&\text{otherwise,}
    \end{cases} 
\end{align*}
for $t=1{:}(T-1)$.
By definition, $\thedelta\ge \theepsilon$, and therefore
$\frac{\thedelta c}{\thedelta c + N-c} \ge \frac{\theepsilon c}{\theepsilon c +
  N-c} \ge \frac{\theepsilon c}{N}$. This, together with
$|C_{T}|\gest \hat{C}_{T-\kappa}$
implies that $\P(J_T = \tilde{J}_T\in C_T)\ge \P(\xi_{T-\kappa}=1)$.
Similarly, by~\cite[Theorem 6.B.3]{shaked-shanthikumar},
we deduce that
\[
    \big(\charfun{\smash{J_1 = \tilde{J}_1\in C_1}},\ldots,
    \charfun{\smash{J_T =
      \tilde{J}_T\in C_T}}\big) \gest
    \xi_{(1-\kappa):(T-\kappa)}.
\]
Because the functions $\phi_t(x_{1:T}) = \prod_{u=1}^t \max\{0,x_u\}$
are increasing, the claim follows.
\end{proof} 

\begin{lemma}
\label{lem:ic-res} 
Suppose $0<\omega_*\le \omega^*<\infty$ and
$\omega^{(i)},\tilde{\omega}^{(i)}\in[\omega_*,\omega^*]$
for $i\in\{1{:}N\}$. Let
\[
\varepsilon \defeq  \frac{\omega_*}{\omega^*}
\qquad\text{and}\qquad
C \defeq \{j\in\{1{:}N\}\given
          \omega^{(j)}=\tilde{\omega}^{(j)}\}.
\]
Then, $(I^{(1:n)},\tilde{I}^{(1:n)}) \sim
\textsc{CRes}(\omega^{(1:N)}, \tilde{\omega}^{(1:N)}, n)$ satisfy the
following for all $j\in \{1{:}n\}$:
\begin{enumerate}[(i)]
    \item \label{item:ic-simple}
      $\P(I^{(j)}=\tilde{I}^{(j)}=i) \ge
      \frac{\varepsilon}{N}$
      for all $i=1{:}N$,
    \item \label{item:ic-coupled}
      $\P(I^{(j)}=\tilde{I}^{(j)}\in C)
    \ge \frac{|C| \varepsilon}{|C|\varepsilon + N-|C|}$.
\end{enumerate}
\end{lemma}
\begin{proof} 
Note that $\P(I^{(j)}=\tilde{I}^{(j)}=i) = w^{(i)}\wedge
\tilde{w}^{(i)}$, so
the first bound is immediate.
For the second, let $C^c \defeq
\{1{:}N\}\setminus C$, and observe that
\begin{align*}
    \sum_{j\in C} w^{(j)} \wedge \tilde{w}^{(j)}
    &= \frac{\sum_{j\in C} \omega^{(j)}}{\sum_{i\in C} \omega^{(i)} +
      \big(\sum_{i\in C^c} \omega^{(i)}\big) \vee
      \big(\sum_{i\in C^c} \tilde{\omega}^{(i)}\big)}
    \\
    &\ge \frac{ |C| \omega_*
       }{
      |C| \omega_*  + |C^c| \omega^*},
\end{align*}
because $x\mapsto x(x+b)^{-1}$ is increasing for $x\ge 0$ for any $b>0$.
The last bound equals~\eqref{item:ic-coupled}.
\end{proof}

Because $\kappa_{n+1} - \kappa_n \gest \Delta \wedge
(T-\kappa_n)$, we note that $\tau\lest \hat{\tau}$,
where
\begin{equation}
    \hat{\tau} \defeq \inf\bigg\{ n\ge 0 \given \sum_{k=1}^n \Delta_k
      \ge T \bigg\},
    \label{eq:tau-hat}
\end{equation}
and $\Delta_k$ are independent realisations of $\Delta$ in
Lemma~\ref{lem:domination}. The next lemma indicates that if $N$ is
large enough (given $\thedelta,\theepsilon$), the random variables $\Delta$
are well-behaved, and ensure good expectation and tail probability
bounds for $\hat{\tau}$.

\begin{lemma}
\label{lem:drift} 
Given any $N\in\N$,
consider the random variable $\Delta$ defined in Lemma~\ref{lem:domination}.
For any $\beta<\infty$ and $\alpha\in\big(1,
(1-\theepsilon)^{-1}\big)$,
there exists
$N_0<\infty$ such that for all $N\ge N_0$,
\[
    \E[\Delta] \ge \beta \qquad\text{and}\qquad
    \E[\alpha^{-\Delta}] \le \beta^{-1}.
\]
\end{lemma}
\begin{proof} 
Suppose $\varphi:\N\to\R_+$ is decreasing and $L\in\N$, then
\begin{align*}
    \E[\varphi(\Delta)\mid \hat{C}_{1:s}] &=
    \sum_{t=-\infty}^s \P(\xi_{-\infty:t}=1,\xi_{t+1}=0\mid \hat{C}_{1:s}) \varphi(t) \\
    &\le \sum_{t=-\infty}^{L\wedge s} \P(\xi_{-\infty:t}=1,\xi_{t+1}=0\mid \hat{C}_{1:s}) \varphi(t)
    + \varphi(L),
\end{align*}
and because $p_u^{(1)}=1$ and $p_u^{(0)} = \theepsilon$
for $u\le 0$, we may write
\begin{align*}
    \P(\xi_{-\infty:t}=1,\xi_{t+1}=0&\mid \hat{C}_{1:s}) \\
        &= \bigg[\prod_{u=1}^{t-1} p_u^{(1)} \bigg] p_t^{(0)}
        \bigg[\prod_{u=t+1}^{0} (1-\theepsilon)\bigg]
       \P(\xi_{(t+1)\vee 1}=0\mid \hat{C}_{1:s}).
\end{align*}
Furthermore, for $t\in\{1{:}L\}$,
\begin{align*}
    \P(\xi_t=0\mid \hat{C}_{1:s})
    &= \P(\xi_{t:(L+1)}=0\mid \hat{C}_{1:s})
    + \hspace{-4ex}\sum_{b\in \{0,1\}^{L-t},\,b\neq 0}\hspace{-4ex}
    \P(\xi_t = 0, \xi_{t+1:(L+1)}=b\mid \hat{C}_{1:s}) \\
    &\le \prod_{u=t}^{L} (1-p_u^{(0)})
    + \sum_{u=t}^L (1-p_u^{(1)}).
\end{align*}
Lemma~\ref{lem:binom-expectation} implies that for $t=1{:}L$,
the terms
$R_t \defeq \hat{C}_t/N \to 1$
in probability as $N\to\infty$, and consequently also
$p_t^{(1)} \to 1$
and $p_t^{(0)}\to \theepsilon$.
We conclude that whenever $\sum_{t<0} (1-\theepsilon)^t \varphi(t)<\infty$,
\[
    \limsup_{N\to\infty}
    \E[\varphi(\Delta)] \le \varphi(L) + \sum_{t=-\infty}^{L}
    (1-\theepsilon)^{L-t}\theepsilon \varphi(t)
    = \varphi(L) + \sum_{t=0}^\infty (1-\theepsilon)^t \theepsilon \varphi(L-t).
\]
We get the first bound by
and applying the result above with
$\varphi(t) = (L-t)_+$, because $\E[\Delta] \ge L - \E[(L-\Delta)_+]$,
and $\limsup_{N\to\infty} \E[(L-\Delta)_+]\le 1-\theepsilon^{-1}$.
The second bound follows by taking $\varphi(t) = \alpha^t$, because
\[
    \limsup_{N\to\infty}
    \E[\varphi(\Delta)] \le \alpha^{-L} +
    \sum_{t=0}^\infty (1-\theepsilon)^t \theepsilon \alpha^{t-L}
    = \alpha^{-L} \big[ 1 + \theepsilon \big( 1 - (1-\theepsilon)\alpha
    \big)^{-1}\big].\qedhere
\]
\end{proof} 

\begin{lemma}
\label{lem:binom-expectation} 
The expectation of $\hat{C}_t$ generated in Lemma~\ref{lem:domination} may
be lower bounded as follows:
\[
    \E\bigg[\frac{\hat{C}_t}{N}\bigg] \ge \frac{\delta_N^{t-1}(N-1)}{1
      + \delta_N^{t-1}(N-1)},
    \qquad\text{where}\qquad
    \delta_N \defeq \frac{N-1}{N}
\thedelta.
\]
Therefore, for any $t\in\N$ and $\varepsilon>0$, there exists $N_0$
such that for all $N\ge N_0$ and all $u=1{:}t$,
$\E[\hat{C}_u/N]\ge 1-\varepsilon$.
\end{lemma} 
\begin{proof} 
Denote $R_t \defeq \hat{C}_t/N$,
then
for any $t\ge 1$, we have
\[
    \E[R_t\mid R_{t-1}]
      = \frac{\delta_N R_{t-1}
        }{ 1-
        (1-\thedelta) R_{t-1}  }
\]
Note that for $a,b>0$ and $\lambda\in[0,b)$, the function
$x\mapsto ax(b-\lambda x)^{-1}$
is convex on $[0,1]$. Therefore, by Jensen's inequality,
\begin{align*}
    \E\bigg[ \frac{a_t R_t}{b_t - \lambda_t R_t}
       \,\bigg|\, R_{t-1}\bigg]
    &\ge \frac{a_t \E[R_t\mid R_{t-1}] }{b_t -
      \lambda_t \E[R_t\mid R_{t-1}]
      } \\
    &= \frac{a_t \delta_N R_{t-1} }{b_t
      [1 -(1-\thedelta) R_{t-1}] - \lambda_t \delta_N R_{t-1}} \\
    &= \frac{a_{t-1} R_{t-1}}{b_{t-1} - \lambda_{t-1}
      R_{t-1}},
\end{align*}
where $a_{t-1} = \delta_N a_t$, $b_{t-1} = b_t$ and
$\lambda_{t-1} = \delta_N \lambda_{t} +
    (1-\thedelta) b_t$.
Starting with $a_t = 1$, $b_t = 1$ and $\lambda_t=0$, we conclude that
\[
    a_1 = \delta_N^{t-1}, \qquad
    b_1 = 1, \qquad\text{and}\qquad
    \lambda_1 = (1-\thedelta)\sum_{k=0}^{t-2} \delta_N^k
    = (1-\thedelta) \frac{1 - \delta_N^{t-1}}{1 - \delta_N},
\]
and consequently,
\[
    \E[R_t] \ge \frac{a_1 R_1}{b_1 + \lambda_1 R_1}
    = \frac{\delta_N^{t-1} \frac{N-1}{N}}{1 -
      (1-\thedelta) \frac{1 - \delta_N^{t-1}}{1 - \delta_N} \frac{N-1}{N}}
    \ge \frac{\delta_N^{t-1} \frac{N-1}{N}}{1 -
      (1 - \delta_N^{t-1}) \frac{N-1}{N}},
\]
because $(1-\delta_N)>(1-\thedelta)$.
This equals the desired bound.
\end{proof} 

\begin{remark} 
In order to make the bound in Lemma~\ref{lem:binom-expectation} large,
we must have $\delta_N^{t-1}(N-1)\gg 1$. Because $\delta_N^{t-1}\ge
(1-t/N)\thedelta^{t-1}$, it is sufficient that we take $N\gg t$ and
$\log(N) \ge c+ t(-\log\thedelta)$. Usually the latter is dominant,
meaning that $N$ of order $\thedelta^{-t}$ is necessary.

We simulated the random variables $\hat{C}_t$, and observed a similar
`cutoff' --- a $\thedelta^{-1}$-fold increase in $N$ caused the
`boundary' where $\hat{C}_t/N$ starts to drop from zero around one to
zero, to extend by one step further. We believe that Lemma~\ref{lem:binom-expectation} captures the behaviour of $\hat{C}_t$
rather accurately. However, we believe that $\hat{C}_{t-\kappa}$ are
often rather pessimistic compared with the actual couplings $|C_t|$.
\end{remark} 

\begin{proof}[Proof of Theorem~\ref{thm:ccbpf-coupling-time}] 
The result follows from the fact that $\P(\tau \ge n) \le
\P(\hat{\tau} \ge n)$ where $\hat{\tau}$ is given in~\eqref{eq:tau-hat},
and a Chernoff bound
\begin{align*}
    \P(\hat{\tau} \ge n)  \le
    \P\bigg(\sum_{i=1}^n \Delta_i \le T\bigg)
     \le \min_{u>0} e^{uT} \prod_{i=1}^n \E[e^{-u \Delta_i}]
    \le \hat{\alpha}^T \beta^{-n},
\end{align*}
where the final inequality uses Lemma~\ref{lem:drift} by choosing
$u=\log\hat{\alpha}$, for some $\hat{\alpha} \le \alpha$ such that
$\hat{\alpha}\in\big(1,(1-\theepsilon)^{-1}\big)$.
\end{proof} 


\section{Empirical comparison}
\label{sec:empirical} 

We compare the CCBPF with the CCPF, as well as the CCPF with ancestor sampling
proposed by~\cite{jacob-lindsten-schon} that was found to outperform the CCPF in their experiments.
Like the CCBPF, the CCPF with ancestor sampling is a coupling of the CBPF.
To emphasize the main difference between the algorithms, we abbreviate the
variants according to the way ancestors are sampled as declared in lines 15 and 17 of Algorithm~\ref{alg:ccxpf}.
The CCPF is denoted AT (ancestor tracing),
the CCBPF is denoted BS (backward sampling),
and the CCPF with ancestor sampling is denoted AS (ancestor sampling).

The computational cost of Algorithm~\ref{alg:unbiased} is the coupling time $\tau$
multiplied by the cost of each run of the CC\textsc{x}PF, at least when the burn-in $b$ is small.
We refer to this as the total cost of coupling.
For AT, AS and BS the cost per iteration is $O(NT)$, and AS and BS are often
around twice the cost of AT.
Hence the expected total cost of coupling is $\mathbb{E}[\tau]NT$, up to some
constant depending on the model, with an additional factor of about 2 for AS and BS.
We investigate below the dependence of $\tau$ and $\mathbb{E}[\tau]$ on $N$ and $T$ for two models:
\begin{enumerate}
\item the linear Gaussian model used by~\cite{jacob-lindsten-schon} with the same data;
\item a simple homogeneous model where the transitions correspond to a simple random walk model with standard normal increments, and the potential functions are $G_t(x_t)=\mathbf{1}_{[-s,s]}(x_t)$ for $t \in \{1,\ldots,T\}$, with either $s=5$ or $s=10$.
\end{enumerate}
To give an idea of computational cost in seconds, the cost of running AT, AS and BS with
$(T,N) = (1000,1024)$ for the linear Gaussian model is respectively approximately
$100$ms, $160$ms and $160$ms on both a Xeon E5-2667v3 CPU and a 2018 Macbook Air.

For fairer comparison, we report results with proposals using common random
numbers to execute line~\ref{line:uncoupled-mutation} of Algorithm~\ref{alg:ccxpf}, as suggested in \citep{jacob-lindsten-schon}.
Our implementation of AT and AS differs from theirs only by minor modifications
that empirically have no substantive effect on any important characteristics of
the algorithm; for instance they quantile-couple the residual indices (cf. line 7 of
Algorithm~\ref{alg:cres}).
In Appendix~\ref{sec:supp_sims} we report the results of similar experiments run
without common random numbers in which both AT and AS
perform considerably worse, while BS is only slightly affected.
This suggests that accurate analysis of AT or AS may require analysis of the effect
of common random numbers which has thus far not been undertaken here or in \citet{jacob-lindsten-schon}.

In Figures~\ref{fig:mean-coupling-times} (a)--(c) we plot the mean coupling times
from 1000 replications using the linear Gaussian model, for every combination of
$T \in \{50,100,200,400,800,1600,3200\}$ and $N \in \{64, 128, 256, 512, 1024\}$.
If at least one replication resulted in a coupling time above $2000$, the results
for that $(T,N)$ combination are excluded.
For AT, we see that $N$ needs to grow roughly linearly in $T$ to
maintain the same mean coupling time, and that smaller values of $N$
do not lead to successful coupling in a reasonable amount of time: for
$T \in \{1600, 3200\}$, even $N=1024$ was insufficient to have a reasonable mean
coupling time.
AS and BS exhibit somewhat similar characteristics, with BS having a smaller
mean coupling time for $T$ large.
For $N \in \{64, 128\}$, AS did not have a reasonable mean coupling time for large $T$, unlike BS.
We complement Figures~\ref{fig:mean-coupling-times} (a)--(c) with
Table~\ref{tab:jls-couplingtimes} for combinations of $(T,N)$ in which all algorithms
are somewhat competitive. The table indicates that in this regime,
BS coupling times have a smaller variance.

In Figure~\ref{fig:mean-coupling-times} (d) we plot the mean coupling times for
all of the approaches for the simple homogeneous model with $s=10$.
We used all combinations of
$T \in \{500,1000,2000,4000\}$ and $N \in \{16, 32, 64, 128\}$, with results for
a combination excluded if a replication resulted in a coupling time above $10T$.
This is a homogeneous model, and we see that BS has a mean
coupling time that is roughly linear in $T$.
In contrast, the mean coupling time appears to grow superlinearly for AS.
We observed also that for large $T$, BS did not result in a steady
increase of the coupling boundary for $N \leq 32$.
This suggests that there is indeed a threshold value of $N_0$ necessary for
linear in $T$ mixing: it seems one cannot take $N_0 = 2$ universally in
Theorem~\ref{thm:ccbpf-coupling-time} with an appropriately large $\rho$.
While this model is relatively simple, it is challenging for coupled
conditional particle filters because the potential functions are not very
informative about the location of a particle in space.

The two examples show quite different behaviour for AS: in the linear Gaussian model
it appears that AS may have linear-in-time convergence for the range of $T$ considered,
even if its performance appears to be worse than that of BS.
However, for the simple but challenging homogeneous model AS appears not to enjoy
this property. We suspect that the good behaviour of AS in the former case is due
to the combination of using common random numbers (cf. Figure~\ref{fig:mean-coupling-times-ind})
and the fact that the potential functions are highly informative about particle locations.
It is not clear that this type of behaviour will extend to more challenging scenarios, such as
when the state space is higher dimensional or the smoothing distribution is multimodal.

\begin{figure}
    \centering
    \begin{subfigure}[b]{0.49\linewidth}
    \centering
    \includegraphics[width=\linewidth]{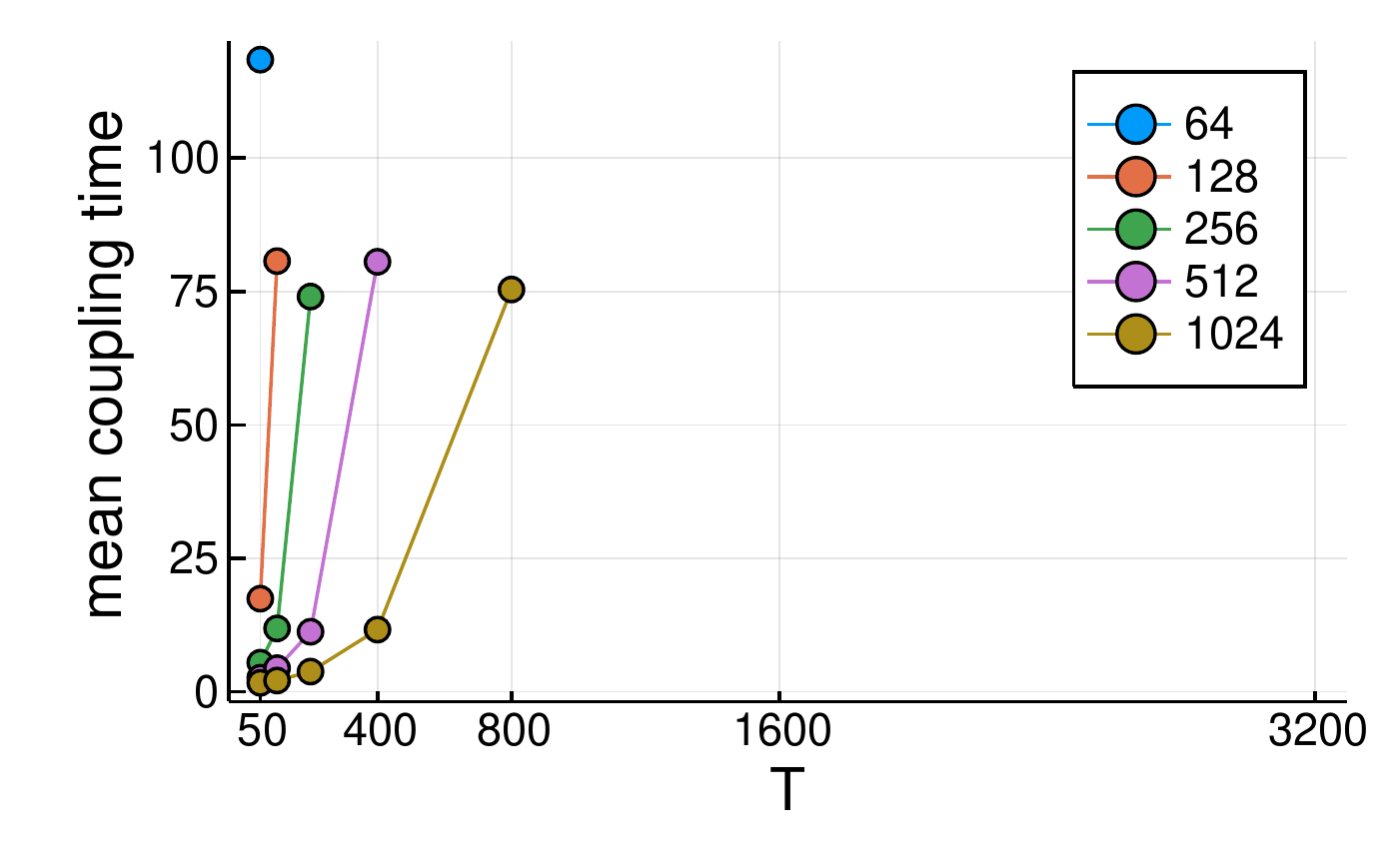}
    \caption{Linear Gaussian: AT}
    \end{subfigure}
    \begin{subfigure}[b]{0.49\linewidth}
    \centering
    \includegraphics[width=\linewidth]{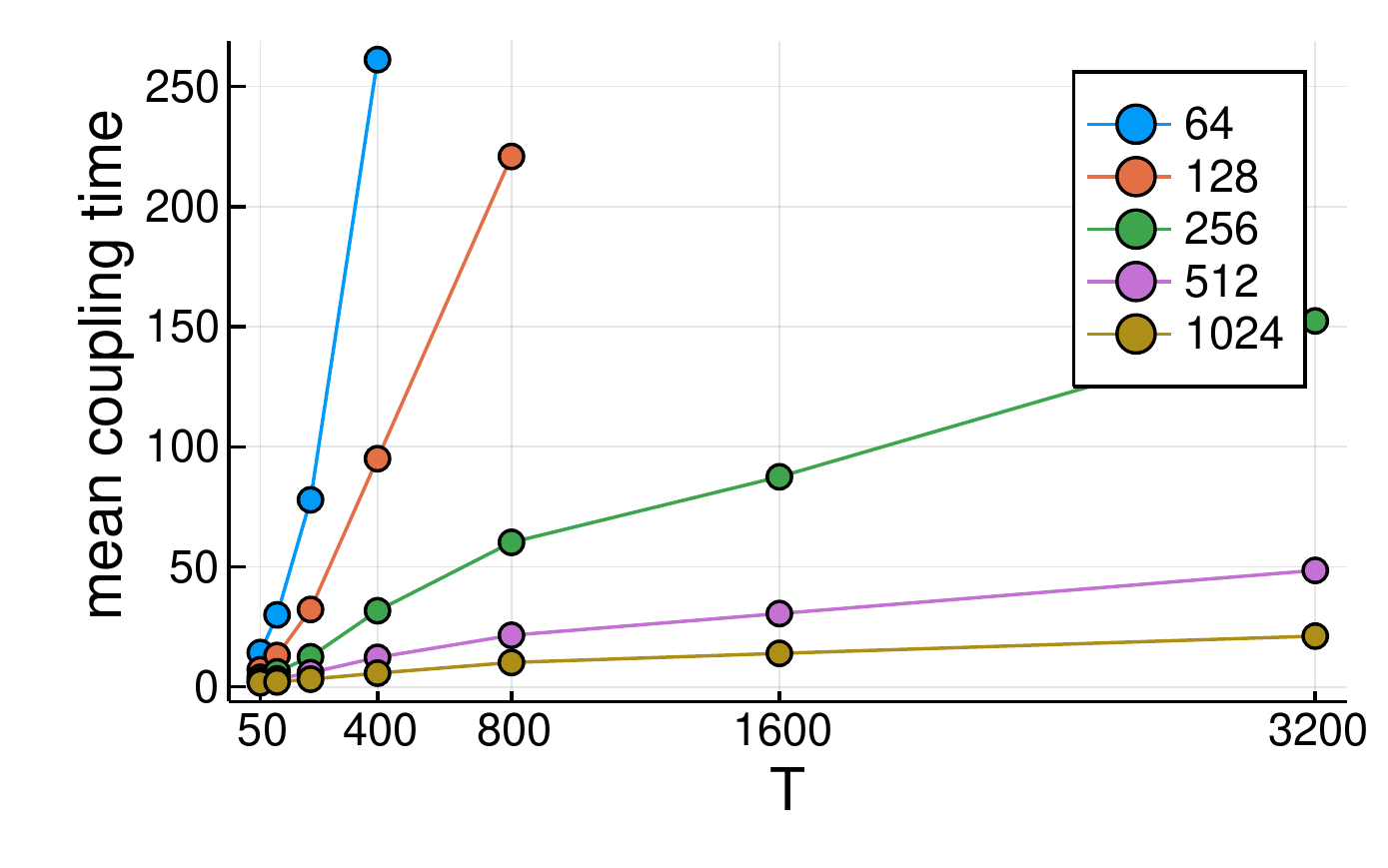}
    \caption{Linear Gaussian: AS}
    \end{subfigure}
    \begin{subfigure}[b]{0.49\linewidth}
    \centering
    \includegraphics[width=\linewidth]{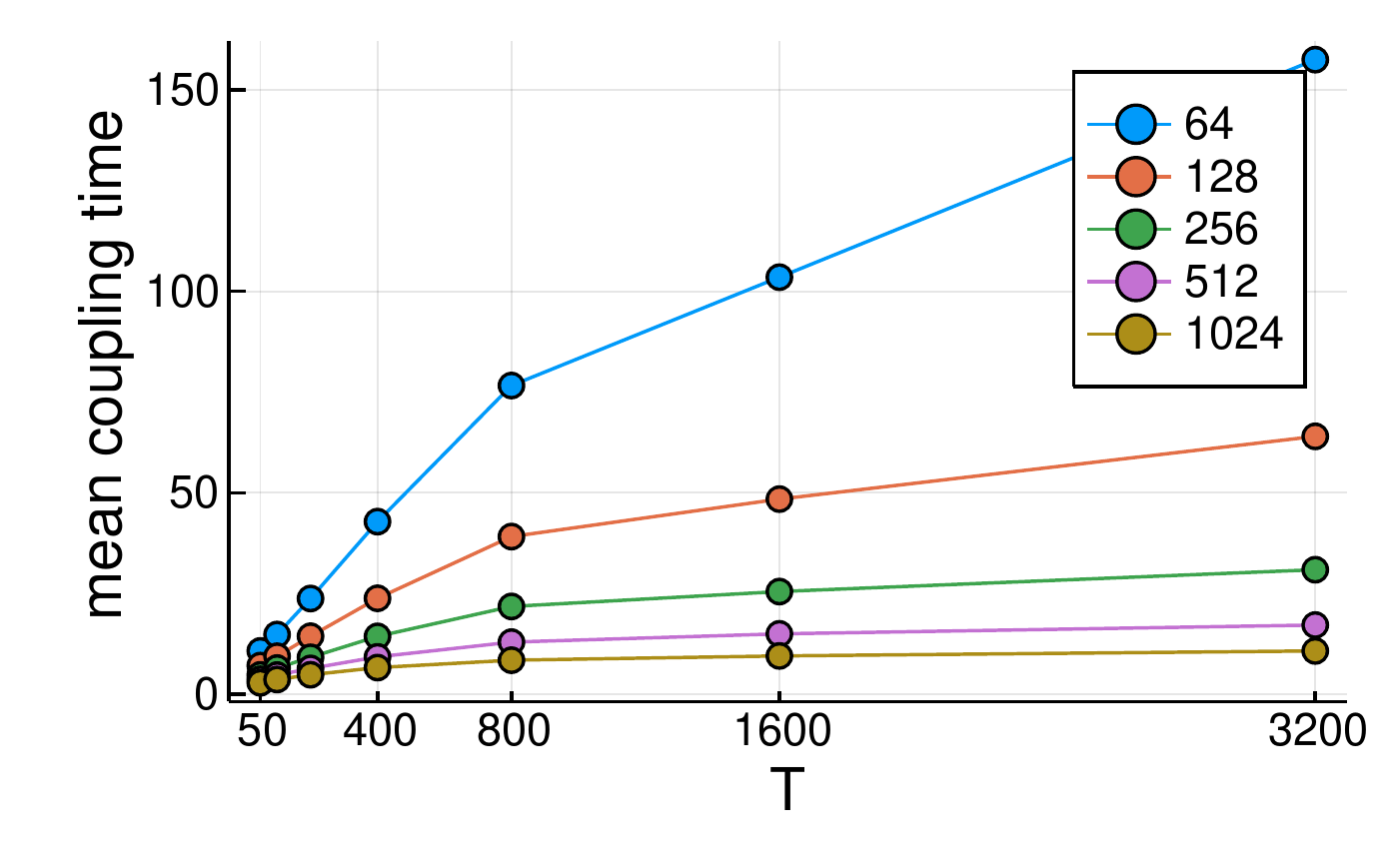}
    \caption{Linear Gaussian: BS}
    \end{subfigure}
    \begin{subfigure}[b]{0.49\linewidth}
    \centering
    \includegraphics[width=\linewidth]{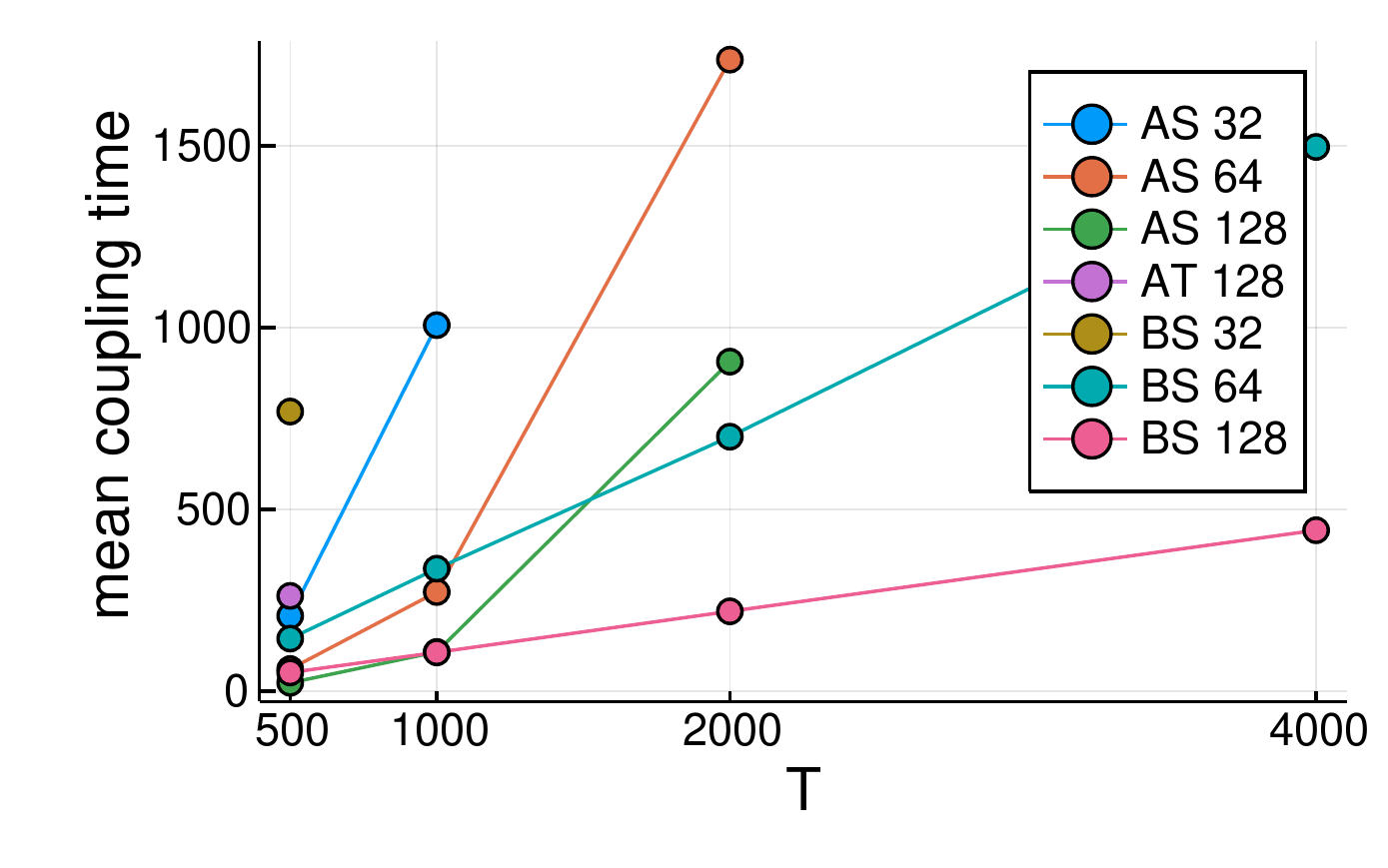}
    \caption{Simple homogeneous model}
    \end{subfigure}
    \caption{Mean coupling times associated with ancestor tracing (AT), ancestor
      sampling (AS) and backward sampling (BS). For (d), the lines are coloured
      according to the type of algorithm and the number of particles $N$.}
    \label{fig:mean-coupling-times}
\end{figure}

\begin{table}
\caption{Average of 1000 coupling times (with standard deviations), with different
  variants of coupled ancestor tracing (AT), ancestor sampling (AS)
  and backward sampling (BS).}
\label{tab:jls-couplingtimes} 
\setlength{\tabcolsep}{2pt}
\scriptsize
\begin{tabular}{ll@{\hskip2pt}ll@{\hskip2pt}ll@{\hskip2pt}ll@{\hskip2pt}ll@{\hskip2pt}ll@{\hskip2pt}ll@{\hskip2pt}ll@{\hskip2pt}l}
\toprule
$T$    & \multicolumn{4}{c}{50}
       &  \multicolumn{4}{c}{100}
       & \multicolumn{4}{c}{200}
       & \multicolumn{4}{c}{400} \\
       \cmidrule(lr){2-5}
       \cmidrule(lr){6-9}
       \cmidrule(lr){10-13}
       \cmidrule(lr){14-17}
$N$    & \multicolumn{2}{c}{64} & \multicolumn{2}{c}{128}
       & \multicolumn{2}{c}{128} & \multicolumn{2}{c}{256}
       & \multicolumn{2}{c}{256} & \multicolumn{2}{c}{512}
       & \multicolumn{2}{c}{512} & \multicolumn{2}{c}{1024} \\
\midrule
AT     & 122.3 & (131.2)
       & 17.3 & (17.1)
       & 77.3 & (82.0)
       & 12.3 & (11.2)
       & 68.2 & (67.5)
       & 10.9 & (9.6)
       & 81.5 & (76.6)
       & 11.7 & (9.9)\\
AS     & 14.2 & (11.0)
       & 7.2 & (5.9)
       & 13.0 & (10.4)
       & 6.3 & (4.5)
       & 12.2 & (8.8)
       & 5.9 & (4.1)
       & 12.5 & (8.2)
       & 5.9 & (3.5) \\
BS     & 11.0 & (5.2)
       & 6.9 & (3.0)
       & 9.5 & (3.3)
       & 6.3 & (2.0)
       & 9.2 & (2.5)
       & 6.4 & (1.7)
       & 9.4 & (2.2)
       & 6.6 & (1.6) \\
\bottomrule
\end{tabular}
\end{table}

Figure~\ref{fig:paths} shows coupling boundaries (see Section~\ref{sec:ccbpf-proof})
by iteration of a single run of each method in the linear Gaussian model,
illustrating typical progressive behaviour of the coupling boundary with BS,
in contrast with AS which does not clearly display a drift towards complete coupling,
and AT which makes no progress at all.
The BS appears viable with much smaller number of particles, and
suggests that the computationally optimal number of particles with BS may differ
significantly from that of AT and AS.

\begin{figure} 
    \begin{center}
        \includegraphics[width=\linewidth]{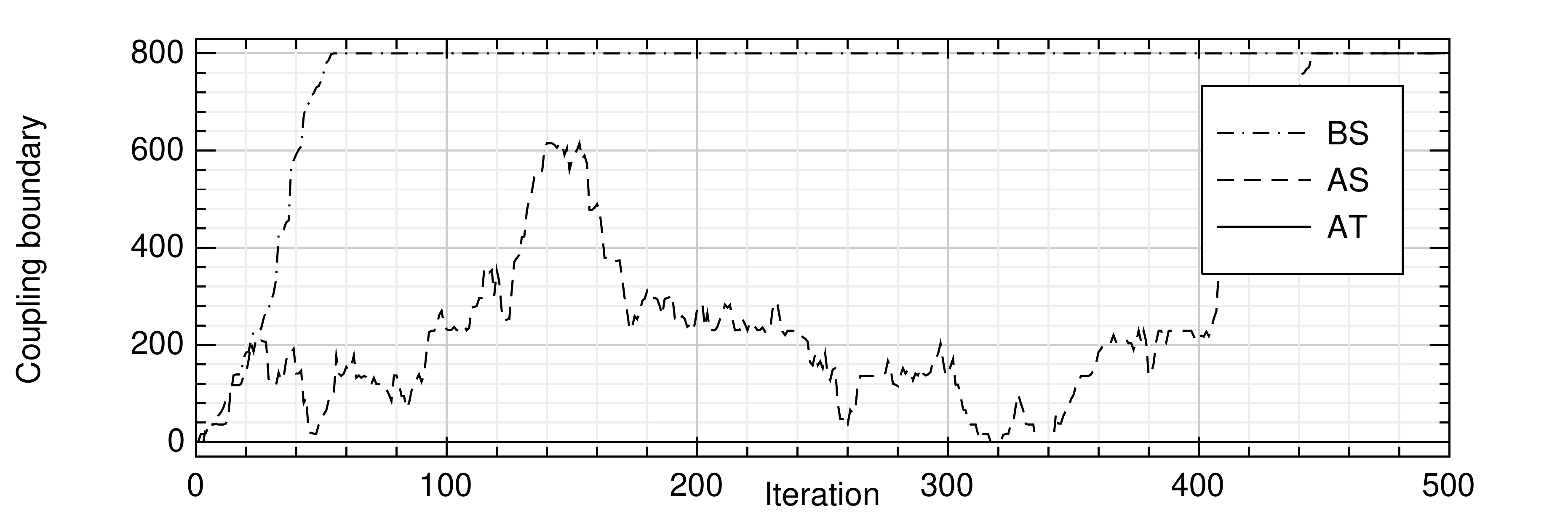}
    \end{center}
    \caption{One realisation of coupling boundaries with $T=800$ and $N=64$.} 
    \label{fig:paths}
\end{figure}

Finally, we compare the total cost of coupling in the simple homogeneous model
with $s=5$, with $N\in\{2^8,2^9,2^{10},2^{11},2^{12},2^{13}\}$
and $T\in\{1000,2000,3000,4000,5000\}$. Figure~\ref{fig:optimalcost}
shows normalised cost of coupling, defined as $N\tau/T$, over 100 replications of both
algorithms. The memory consumption limited the highest number of
particles to $N=8192$ with $T=5000$, which already exceeded 4
gigabytes in our implementation. With the shortest time horizon
$T=1000$, the AS was competitive with BS, reaching sometimes lower
costs than BS. With increasing $T$, the AS failed to couple
increasingly often before reaching the maximum number of iterations $\lfloor 10^9/(NT) \rfloor$,
chosen so that the maximum time spent on a replication was approximately 2 minutes.
The BS only failed to couple before this maximum number of iterations was reached
3 times (out of 100) with $N=8192$ and $T=5000$, and remained effective with small $N$.
This experiment suggests that AS requires $N$ to increase with $T$ in
order to stay effective, leading to a superlinear memory requirement
that may limit its application with longer time horizons.

\begin{figure} 
    \begin{center}
        \includegraphics[width=\linewidth]{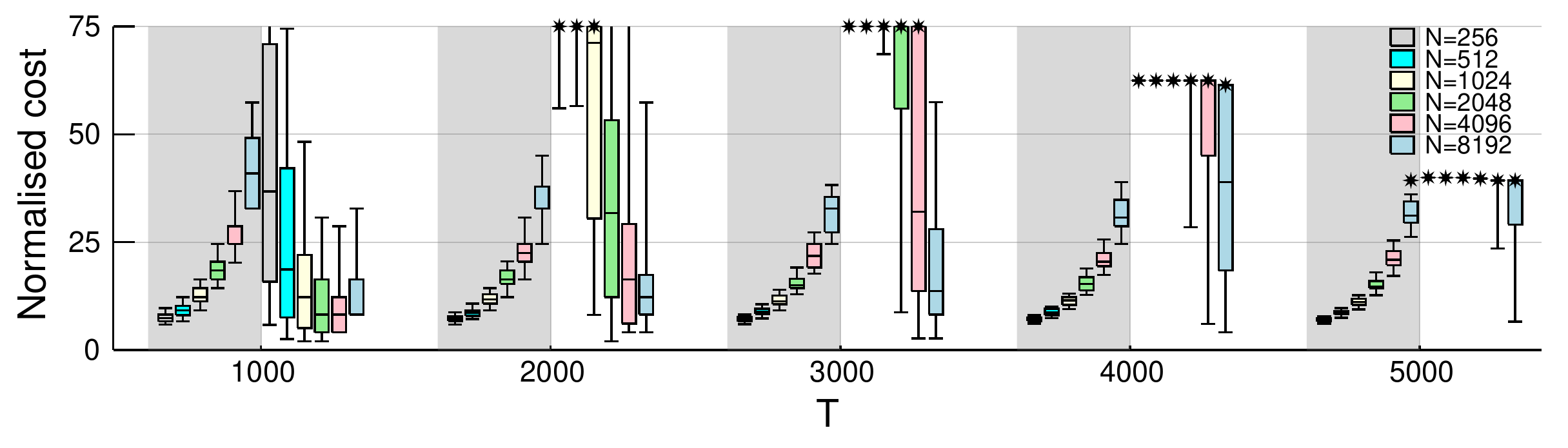}
    \end{center}
    \caption{Normalised cost of coupling for the simple homogeneous model
    for 100 replications of BS and AS, the former on shaded background. The stars indicate
    cases where at least one replication failed to couple before hitting the
    maximum total cost. Since the normalised cost of coupling is plotted, the actual differences
    in computational cost are obtained by scaling by $T^2$.}
    \label{fig:optimalcost}
\end{figure}


To provide finer detail, we report in Figure~\ref{fig:mean-coupling-times-scale}
the results of simulations testing the scaling properties AT, AS and BS for both
the linear Gaussian and simple homogeneous model with $s=5$.
For the linear Gaussian model,
the mean coupling time for AT appears to be stable with $N$ proportional to $T$,
while for both AS and BS the mean coupling time grows roughly linearly with $T$ for
$N$ fixed.
In contrast, for the simple homogeneous model the mean coupling time appears to
grow superlinearly for AT even with $N$ proportional to $T$, linearly for AS
with $N$ proportional to $T$, but linearly for BS with $N$ fixed.
This suggests that taking $N$ proportional to $T$ is not sufficient in general
to stabilize AT even for relatively simple models. Similarly, for the
homogeneous model, the relative cost of AS over BS grows with $T$ and is already
around 8 for $T=1600$.

\begin{figure}
    \centering
    \begin{subfigure}[b]{0.49\linewidth}
    \centering
    \includegraphics[width=\linewidth]{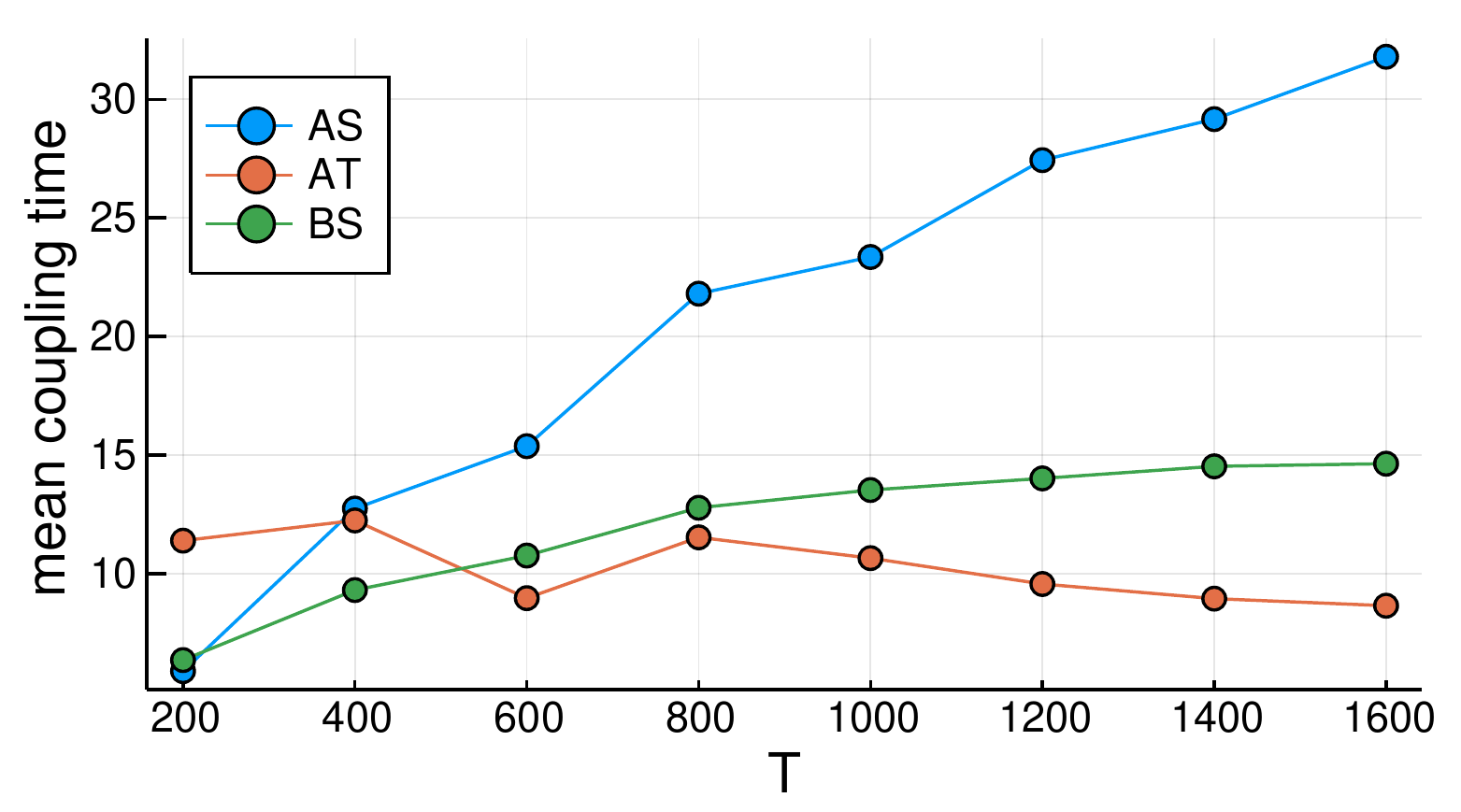}
    \caption{Linear Gaussian model.}
    \end{subfigure}
    \begin{subfigure}[b]{0.49\linewidth}
    \centering
    \includegraphics[width=\linewidth]{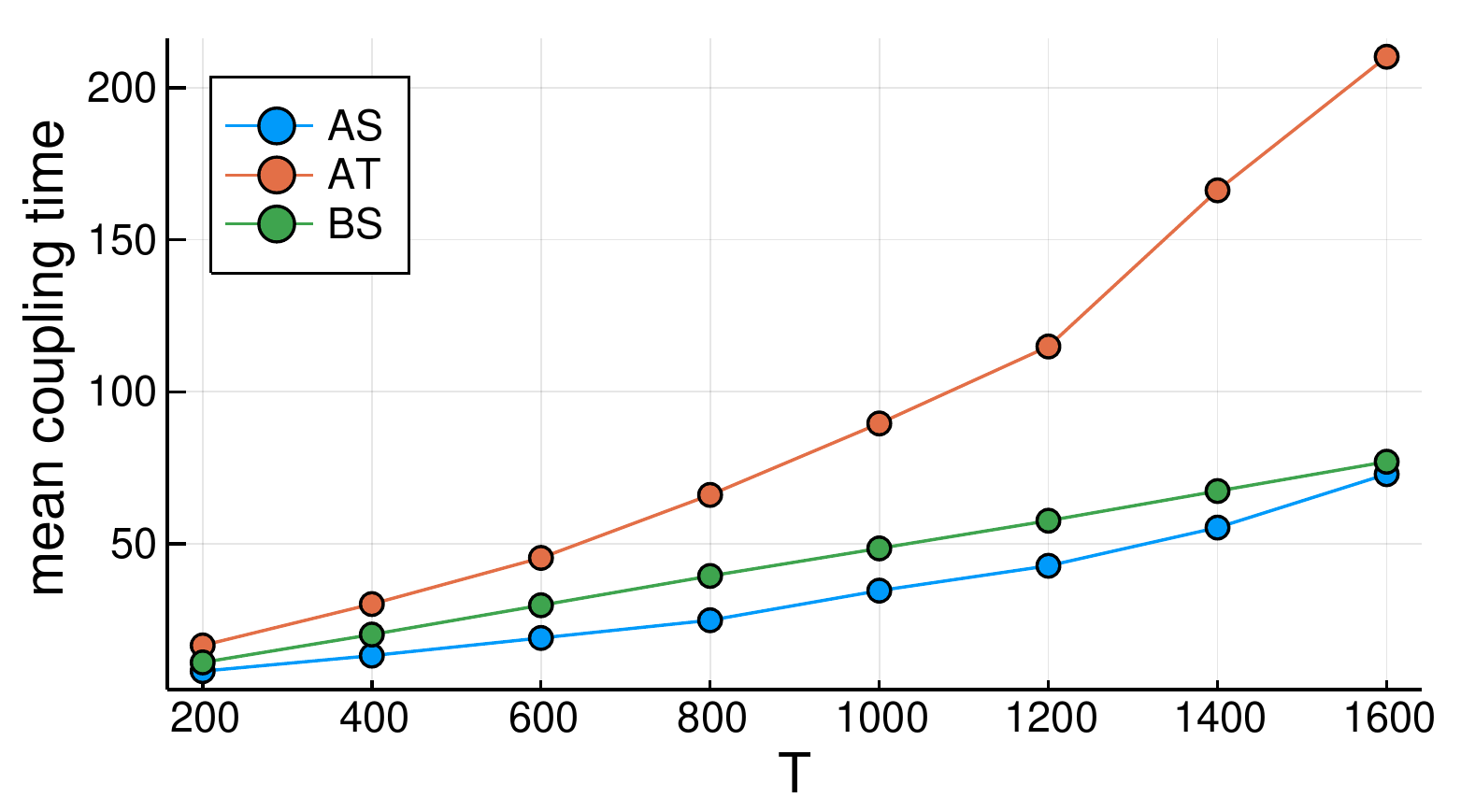}
    \caption{Simple homogeneous model.}
    \end{subfigure}
    \caption{Mean coupling times associated with ancestor tracing (AT), ancestor
      sampling (AS) and backward sampling (BS). For the linear Gaussian model,
      $N = 2.56T$ for AT whereas $N=512$ for AS and BS. For the simple
      homogeneous model, $N = 0.64T$ for AT and AS, whereas $N=128$ for BS.}
    \label{fig:mean-coupling-times-scale}
\end{figure}

Our empirical results suggest that all of AT, AS and BS are competitive when $T$ is
small, and one should take $N$ proportional to $T$ for AT and sometimes also for AS.
Since the space complexity of any algorithm when taking $N$ proportional to $T$ is
quadratic in $T$, this approach does not scale to large values of $T$.
When it is no longer possible due to memory requirements to
take $N$ proportional to $T$, BS appears to be the only appropriate algorithm.
We are not able to quantify accurately the minimal number of particles $N_0$
required for BS to exhibit linear-in-time convergence or the value of $N$ that
maximises its computational efficiency, so this needs to be done empirically.


\section*{Acknowledgements} 
The authors would like to thank the Isaac Newton
Institute for Mathematical Sciences and the Insititute for
Mathematical Sciences at the National University of Singapore for
support and hospitality during the programmes ``Scalable inference;
statistical, algorithmic, computational aspects'' and ``Bayesian
Computation for High-Dimensional Statistical Models'' respectively,
when work on this paper was undertaken. This work was supported by
EPSRC grant numbers EP/K032208/1, EP/R014604/1 and EP/R034710/1, and by The Alan
Turing Institute under the EPSRC grant EP/N510129/1. MV was supported
by Academy of Finland grants 274740, 284513, 312605 and 315619.
The authors wish to acknowledge CSC, IT Center for
Science, Finland, for computational resources.

\bibliographystyle{abbrvnat}
\bibliography{refs.bib}

\appendix
\section{One-shot coupling probability of CCPF}
\label{app:oneshot} 

\begin{lemma}
\label{lem:convex_sum_rvs} 
Suppose $Y^{(1:n)}$ are non-negative random numbers,
$Z^{(1:n)}$ are $\mathcal{Z}$-valued random variables,
$f:\mathcal{Z}\to[0,b]$ is measurable, and $\mathcal{G}$ is a
$\sigma$-algebra.
If $Y^{(i)}$ are $\sigma(\mathcal{G},Z^{(i)})$ measurable and
$Z^{(1:n)}$ are conditionally independent given $\mathcal{G}$,
then for any $\G$-measurable $A\ge0$,
\[
\E\bigg[\frac{\sum_{i=1}^{n}Y^{(i)}
  }{
  A + \sum_{j=1}^{n}f(Z^{(j)})}\;\bigg|\;
\mathcal{G}\bigg] \ge
\frac{\sum_{i=1}^{n}\E[Y^{(i)}\mid \mathcal{G}]}{A+ b
  +\sum_{j=1}^{n}\E[f(Z^{(j)})\mid \mathcal{G}]}.
\]
\end{lemma} 
\begin{proof} 
The claim is trivial whenever
$\P\big(A +\sum_{j=1}^{n}f(Z^{(j)})=0\bigmid \G\big)>0$, so consider
the case $A +\sum_{j=1}^{n}f(Z^{(j)})> 0$.
Because $x\mapsto x^{-1}$ is convex on $(0,\infty)$,
\begin{alignat*}{1}
    \E\bigg[\frac{Y^{(i)}}{A+\sum_{j} f(Z^{(j)})}\;\bigg|\;\mathcal{G}\bigg]
 & \ge
 \E\bigg[
 \frac{Y^{(i)}}{A+Z^{(i)}+\sum_{j\neq i}\E[f(Z^{(j)})\mid
   \mathcal{G},Z^{(i)}]}\;\bigg|\; \mathcal{G}\bigg]\\
 & \ge \frac{\E[Y^{(i)}\mid \mathcal{G}] }{A+b+\sum_{j}
   \E [f(Z^{(j)})\mid \mathcal{G}] },
\end{alignat*}
whence the result follows.
\end{proof} 

\begin{lemma}
    \label{lem:coupled-ccpfs} 
Consider an augmented state space $\bar{\X} = \X\cup\{\phi\}$, and
define
\begin{itemize}
    \item $\bar{G}_t(x) \defeq G_t(x)$
      and $\bar{G}_t(\phi) \defeq \sup_x G_t(x)$
      for all $t=1{:}T$ and $x\in\X$,
    \item $\bar{M}_t(x,A) = M_t(x,A)$ for all $t=1{:}T$, $x\in\X$ and measurable
$A\subset\X$,
    \item $\bar{M}_t(\phi,\{\phi\})=1$ for $t=2{:}T$.
\end{itemize}
Let $\overline{\mathrm{CCPF}}$ and $\mathrm{CCPF}$ stand for
the CCPF
for models $(\bar{\X},\bar{M}_{1:T},\bar{G}_{1:T})$
and $(\X, M_{1:T},G_{1:T})$, respectively. Then, for all
$s,\tilde{s}\in\X^T$,
\begin{enumerate}[(i)]
    \item \label{item:aug-equivalence}
      $\overline{\mathrm{CCPF}}(s,\tilde{s},N) \eqd
      \mathrm{CCPF}(s,\tilde{s},N)$.
\end{enumerate}
Let $C_{1:T}$ stand for the sets generated in
      by $\overline{\mathrm{CCPF}}(s, \tilde{s}, N)$ and
      $C_{1:T}^\phi$ stand for those generated in
      $\overline{\mathrm{CCPF}}\big(s, (\phi,\ldots,\phi), N\big)$.
\begin{enumerate}[(i)]
    \stepcounter{enumi}
    \item \label{item:aug-minor}
      There exists a coupling such that
      $C_{t}\supset C_t^{\phi}$ a.s.~for all $t=1{:}T$.
    \item $C_t^{\phi} = \big\{i\in\{2{:}N\}\given
        \tilde{X}_t^{(i)}\neq \phi\big\}$.
\end{enumerate}
\end{lemma} 
\begin{proof} 
The marginal equivalence~\eqref{item:aug-equivalence} is
straightforward. For the stochastic minorisation~\eqref{item:aug-minor}, we consider running
$\overline{\mathrm{CCPF}}(s, \tilde{s}, N)$ and
$\overline{\mathrm{CCPF}}\big(s, (\phi,\ldots,\phi), N\big)$
simultaneously, in a coupled manner. More specifically, set
\[
    X_1^{(2:N)} = \tilde{X}_1^{(2:N)} = X_1^{\phi(2:N)} =
\tilde{X}_1^{\phi(2:N)}
\sim M_1(\uarg).
\]
For $t\ge 2$, we proceed inductively, assuming that $X_{t-1}^{(i)} =
\tilde{X}_{t-1}^{(i)} = X_{t-1}^{\phi(i)} = \tilde{X}_{t-1}^{\phi(i)}$
for all $i\in C_{t-1}^{\phi}$, and that $C_{t-1}\supset C_{t-1}^\phi$,
which obviously hold for $t=2$.
Note that then
\begin{align*}
    \omega_{t-1}^{(i)} = \tilde{\omega}_{t-1}^{(i)} &=
    \omega_{t-1}^{\phi(i)} = \tilde{\omega}_{t-1}^{\phi(i)},
    & i&\in C_{t-1}^\phi \\
    \omega_{t-1}^{(i)} \vee \tilde{\omega}_{t-1}^{(i)} &\le
    \tilde{\omega}_{t-1}^{\phi(i)},
    & i&\notin C_{t-1}^{\phi}.
\end{align*}
Also, $\omega_{t-1}^{\phi(i)} \le \tilde{\omega}_{t-1}^{\phi(i)}$ for
$i\notin C_{t-1}^{\phi}$, so we conclude that
\[
    \frac{\omega_{t-1}^{(i)}}{\sum_j \omega_{t-1}^{(j)}} \wedge
    \frac{\tilde{\omega}_{t-1}^{(i)}}{\sum_j \tilde{\omega}_{t-1}^{(j)}}
    \ge
    \frac{\tilde{\omega}_{t-1}^{\phi(i)}}{\sum_j
      \tilde{\omega}_{t-1}^{\phi(j)}}
    = \frac{\omega_{t-1}^{\phi(i)}}{\sum_j
      \omega_{t-1}^{\phi(j)}}
    \wedge
    \frac{\tilde{\omega}_{t-1}^{\phi(i)}}{\sum_j
      \tilde{\omega}_{t-1}^{\phi(j)}},\qquad i \in C_{t-1}^{\phi}.
\]
Consequently, the outputs of \textsc{CRes} satisfy
$\P(I_t^{(i)} = \tilde{I}_t^{(i)} \in
    C_{t-1}^\phi) \ge \P (I_t^{\phi(i)} = \tilde{I}_t^{\phi(i)} \in
    C_{t-1}^\phi)$, and we may couple the outputs such that
\[
    \P(I_t^{(i)} = \tilde{I}_t^{(i)} =
    I_t^{\phi(i)} = \tilde{I}_t^{\phi(i)} \in
    C_{t-1}^\phi) = \P (I_t^{\phi(i)} = \tilde{I}_t^{\phi(i)} \in
   C_{t-1}^\phi),
\]
and consequently we may also couple
$X_t^{(i)},\tilde{X}_t^{(i)},X_t^{\phi(i)},\tilde{X}_t^{\phi(i)}$ such
that
\[
    X_t^{(i)} = \tilde{X}_t^{(i)} = X_t^{\phi(i)} = \tilde{X}_t^{\phi(i)}
    \sim M_t(X_{t-1}^{(I_t^{(i)})},\uarg),
    \qquad i\in C_t^{\phi}.
    \qedhere
\]
\end{proof} 

\begin{proof}[Proof of~\ref{thm:oneshot-bounds}] 
Consider $\overline{\mathrm{CCPF}}\big(s, (\phi,\ldots,\phi), N\big)$,
let $\check{C}_t \defeq \big\{i\in\{2{:}N\}\given \tilde{X}_t^{(i)}\neq \phi\}$,
${\xi}_{t}=\sum_{i=1}^{N}\delta_{\tilde{X}_{t}^{(i)}}$,
${\xi}_{\check{C}_{t}}=\sum_{i\in \check{C}_{t}}\delta_{\tilde{X}_{t}^{(i)}}$,
then by Lemma~\ref{lem:coupled-ccpfs}
\begin{align*}
    \P \big( X_{1:T}^{(J_{1:T})} = \tilde{X}_{1:T}^{(\tilde{J}_{1:T})} \big)
    =
    \P\big( J_T=\tilde{J}_T \in C_T\big)
    \ge \E \bigg[ \frac{{\xi}_{\check{C}_T}(G_T) }{{\xi}_T(G_T)}\bigg].
\end{align*}
Note that the latter quantity does not depend on $X_{t}^{(i)}$, but
only on the marginal conditional particle filter $\tilde{X}_t^{(i)}$ with reference
$(\phi,\ldots,\phi)$.
Setting $h_T^{(1)} \defeq h_T^{(2)}\defeq G_T$,
we may apply Lemma~\ref{lem:convex_sum_rvs}
with $Z^{(i)} = \tilde{X}_T^{(i+1)}$ and $\G=\G_{T-1}$ where
$\G_t = \sigma(\tilde{X}_u^{(i)}\given u\le t, i=2{:}N)$,
\begin{align*}
\E \bigg[ \frac{{\xi}_{\check{C}_T}(G_T) }{{\xi}_T(G_T)}\bigg]
& = \E \bigg[ \frac{{\xi}_{\check{C}_T}(h_T^{(1)})
  }{{\xi}_T(h_T^{(2)})}\bigg] \\
& \ge  \E\bigg[ \frac{\sum_{i=2}^N \E[ \charfun{\smash{\tilde{X}_T^{(i)}\in
    \check{C}_T }} h_T^{(1)}(\tilde{X}_T^{(i)}) \mid \G_{T-1} ]
}{2 \| h_T^{(2)} \|_\infty + \sum_{j=2}^N \E[ h_T^{(2)}(
  \tilde{X}_T^{(j)}) \mid \G_{T-1} ]}
\bigg] \\
&= \E\bigg[ \frac{(N-1) \frac{\xi_{\check{C}_{T-1}} (G_{T-1}
      M_T h_T^{(1)})}{\xi_{T-1}(G_{T-1})}
    }{2 \| h_T^{(2)} \|_\infty + (N-1) \frac{\xi_{T-1}(G_{T-1}
      M_T h_T^{(2)})}{
      \xi_{T-1}(G_{T-1})
      }
    }
\bigg] \\
&= \E\bigg[ \frac{\xi_{\check{C}_{T-1}} (h_{T-1}^{(1)})
    }{\xi_{T-1}(h_{T-1}^{(2)})
    }
\bigg],
\end{align*}
where $h_{t}^{(1)} \defeq G_t M_{t+1} h_{t+1}^{(1)}$
and $h_t^{(2)} \defeq G_t \big( 2 (N-1)^{-1} \| h_{t+1}^{(2)}\|_\infty +
M_{t+1} h_{t+1}^{(2)}\big)$.
We have $h_t^{(1)} \le h_t^{(2)}$, so we may iterate similarly as above
to obtain
\[
    \E \bigg[ \frac{{\xi}_{\check{C}_T}(G_T) }{{\xi}_T(G_T)}\bigg]
    \ge \E\bigg[ \frac{\xi_{\check{C}_{1}} (h_{1}^{(1)})
    }{\xi_{1}(h_{1}^{(2)})
    }\bigg]
    \ge \frac{h_0^{(1)}
    }{h_0^{(2)}
    },
\]
by Lemma~\ref{lem:convex_sum_rvs}, where $h_0^{(1)},h_0^{(2)}$ are defined as
above, with convention $G_0\equiv 1$.

Denoting $Q_t \defeq G_t M_{t+1}$,
$\bar{Q}_{t,u} \defeq Q_t \cdots Q_{u}$ for $t\le u$ and
$\bar{Q}_{t,t} = I$, we have
$h_0^{(1)} = M_1 \bar{Q}_{1,{T-1}}(G_T)$, and
\begin{align*}
    h_0^{(2)}
    &\le 2(N-1)^{-1} \| h_1^{(2)} \|_\infty + \| M_1 h_1^{(2)} \|_\infty.
\end{align*}
We may bound
\begin{align*}
    \|h_t^{(2)}\|_\infty
    &\le G^* \| h_{t+1}^{(2)} \|_\infty \big(
    1 + 2 (N-1)^{-1} \big), \\
    \|M_1 \bar{Q}_{1,t-1} h_t^{(2)}\|_\infty
    &\le 2 (N-1)^{-1} (G^*)^{t} \|h_{t+1}^{(2)} \|_\infty
    + \| M_1 \bar{Q}_{1,t} h_{t+1}^{(2)} \|_\infty,
\end{align*}
and conclude that
\begin{align*}
    h_0^{(2)} \le h_0^{(1)} + c_1 N^{-1},
\end{align*}
for some $c_1=c_1(G^*,T)\in(0,\infty)$. We conclude the claim with
$c = c_1/h_0^{(1)}$.
\end{proof} 


\section{One-shot coupling with rate}
\label{app:oneshot-rate} 

Our results below hold under the following,
slightly more general strong mixing assumption:
\begin{assumption}
    \label{a:mixing-type} 
For any $t=1{:}T$, define $Q_t(x_{t:t+1}) \defeq
G_t(x_t)M_{t+1}(x_t, x_{t+1})$. There exists a
constant $c_*<\infty$ such that
for all  $1\le u\le t\le T$
\[
    \frac{ \sup_{x_u} \int Q_u(x_{u:u+1}) \cdots Q_{t-1}(x_{t-1:t})
    G_t(x_t) \ud x_{u+1:t} }{
    \inf_{x_{u-1}}
    \int
    M_{u}(x_{u-1:u}) Q_u(x_{u:u+1}) \cdots
      Q_{t-1}(x_{t-1:t}) G_t(x_t) \ud x_{u:t} }
    \le c_*.
\]
\end{assumption}
\begin{lemma}
    \label{lem:mixing-to-mixing-type} 
    Suppose $G_t(x_{t-1},x_t) = G_t(x_t)$
  for all $t\ge 2$. Then,
    Assumption~\ref{a:mixing}
      implies Assumption~\ref{a:mixing-type} with $c_* =
      \theepsilon^{-1}$.
\end{lemma} 
\begin{proof} 
Suppose Assumption~\ref{a:mixing} holds. For any non-negative, bounded test
function $h:\X\to\R$, let $m(h) \defeq \int h(y) \ud y < \infty$.
Assumption~\ref{a:mixing} implies that for any $x\in\X$ and $t\ge 2$,
$M_*(t) m(h) \le \int M_t(x,y) h(y) \ud y
\le M^*(t) m(h)$.
Let $1\le u\le t$, and define
$$\phi_{u+1,t}(x_{u+1}) \defeq \int Q_{u+1}(x_{u+1:u+2}) \cdots Q_{t-1}(x_{t-1:t})
    G_t(x_t) \ud x_{u+2:t},$$
which is non-negative and bounded both from above and away from zero. We may calculate
\begin{align*}
&    \sup_{x_u',x_{u-1}'}
\frac{  \int G_u(x_u') M_{u+1}(x_u',x_{u+1}) Q_{u+1}(x_{u+1:u+2}) \cdots Q_{t-1}(x_{t-1:t})
    G_t(x_t) \ud x_{u+1:t} }{
      \int
    M_{u}(x_{u-1}',x_u) Q_u(x_{u:u+1}) \cdots
      Q_{t-1}(x_{t-1:t}) G_t(x_t) \ud x_{u:t} } \\
&\le \sup_{x_u',x_{u-1}'}
\frac{ G_u(x_u') M^*(u+1)  m(\phi_{u+1,t})  }{
      \big(\int
    M_{u}(x_{u-1}',x_u) G_u(x_u) \ud x_u \big)M_*(u+1)
    m(\phi_{u+1,t})} \\
&\le \frac{G^*(u) M^*(u+1)}{G_*(u) M_*(u+1)}. \qedhere
\end{align*}
\end{proof} 

Consider CCPF in Algorithm~\ref{alg:ccxpf}, and
denote $\xi_{t}=\sum_{i=1}^{N}\delta_{X_{t}^{(i)}}$,
$\xi_{C_{t}}=\sum_{i\in C_{t}}\delta_{X_{t}^{(i)}}$,
$\tilde{\xi}_{t}=\sum_{i=1}^{N}\delta_{\tilde{X}_{t}^{(i)}}$,
$\tilde{\xi}_{C_{t}}=\sum_{i\in C_{t}}\delta_{\tilde{X}_{t}^{(i)}}$,
$\rho_t \defeq \xi_{C_t}(G_t)/\xi_t(G_t)$ and
$\tilde{\rho}_t \defeq \tilde{\xi}_{C_t}(G_t)/\tilde{\xi}_t(G_t)$.

\begin{lemma}
\label{lem:ccpf-diff-bound-1} 
Let $h_2\ge h_1\ge 0$ be functions such that
$h_{2}^{*}\defeq \sup_{x}h_{2}(x)<\infty$, then for $t=1{:}(T-1)$,
\[
\E\bigg[ \frac{\xi_{C_{t+1}}(h_{1})}{\xi_{t+1}(h_{2})}\bigg]
\geq\E\bigg[
\frac{\xi_{C_{t}}(h_{1}')}{\xi_{t}(h_{2}')}\bigg] + \E[\tilde{\rho}_t]
- 1,
\]
where
\begin{alignat*}{1}
h_{1}'(x) & =G_t(x) (M_{t+1} h_{1})(x),\\
h_{2}'(x) & =G_t(x) \big[2(N-1)^{-1}h_{2}^{*}+(M_{t+1}h_{2})(x)\big].
\end{alignat*}
\end{lemma}
\begin{proof} 
It is direct to check that $I_{t+1}^{(j)} \in C_{t}$ and
$\tilde{I}_{t+1}^{(j)}\in C_{t}$ implies $I_{t+1}^{(j)}=\tilde{I}_{t+1}^{(j)}$,
because either $\omega_{t}^{(j)} \le \tilde{\omega}_{t}^{(j)}$
for all $j\in C_{t}$ or
$\omega_{t}^{(j)} \ge \tilde{\omega}_{t}^{(j)}$
for all $j\in C_{t}$.
Therefore, we may write
\[
\E\bigg[ \frac{\xi_{C_{t+1}}(h_{1})}{\xi_{t+1}(h_{2})}\bigg]
= \E \bigg[  \frac{\sum_{i=2}^N h_{1}(X_{t+1}^{(i)})1\big((I_{t+1}^{(i)},
  \tilde{I}_{t+1}^{(i)}) \in C_{t}^2\big)
  }{h_{2}(X_{t+1}^{(1)})+\sum_{j=2}^N h_2(X_{t+1}^{(j)})}\bigg],
\]
and apply Lemma~\ref{lem:convex_sum_rvs}
with $\mathcal{G} = \mathcal{G}_{t}\defeq \big\{
  X_{1:t}^{(1:N)},I_{1:t}^{(1:N)},
  \tilde{X}_{1:t}^{(1:N)},\tilde{I}_{1:t}^{(1:N)}\big\}$,
$Y^{(i)} = h_{1}(X_{t+1}^{(i)})
1\big((I_{t+1}^{(i)},
  \tilde{I}_{t+1}^{(i)}) \in C_{t}^2\big) $,
$Z^{(i)} =
(X_{t+1}^{(i)},\tilde{X}_{t+1}^{(i)},I_{t+1}^{(i)},\tilde{I}_{t+1}^{(i)})$
and $f(x,\tilde{x},i,\tilde{i})=h_2(x)$, yielding
\begin{align*}
    \E\bigg[ \frac{\xi_{C_{t+1}}(h_{1})}{\xi_{t+1}(h_{2})}\bigg]
& \ge \E\bigg[
\frac{\sum_{i=2}^{N}
  \E\big[
  h_{1}(X_{t+1}^{(i)})
  1\big((I_{t+1}^{(i)}, \tilde{I}_{t+1}^{(i)}) \in C_{t}^2\big) \;\big|\;
  \mathcal{G}_t\big]
    }{2h_{2}^{\ast}+ \sum_{j=2}^N \E[h_2(X_{t+1}^{(j)})\mid
\mathcal{G}_t]}\bigg] \\
 & =
 \E\Bigg[
 \frac{
 \sum_{i=2}^{N} \frac{\xi_{C_{t}}(G_t (M_{t+1}h_{1}))}{\xi_{C_{t}}
   (G_{t})}
 \E\big[1\big((I_{t+1}^{(i)}, \tilde{I}_{t+1}^{(i)}) \in C_{t}^2
       \big)\;\big|\; \mathcal{G}_{t}\big]
 }{2h_{2}^{\ast}+(N-1)\frac{\xi_{t}(G_t (M_{t+1}h_{2}))}{\xi_{t}(G_{t})}}\Bigg]
\\
 & \geq
  \E\Bigg[
 \frac{
 (N-1) \frac{\xi_{C_{t}}(G_t (M_{t+1}h_{1}))}{\xi_{C_{t}}
   (G_{t})}
 \rho_t \tilde{\rho}_t
 }{2h_{2}^{\ast}+(N-1)\frac{\xi_{t}(G_t
   (M_{t+1}h_{2}))}{\xi_{t}(G_{t})}}\Bigg] \\
 & = \E\Bigg[
 \frac{
 \xi_{C_{t}}(G_t (M_{t+1}h_{1}))
 }{2h_{2}^{\ast}(N-1)^{-1}\xi_t(G_t) + \xi_{t}(G_t
   (M_{t+1}h_{2}))} \tilde{\rho}_t \Bigg]  \\
 &= \E\bigg[ \frac{\xi_{C_t}(h_1')}{\xi_t(h_2')} - (1-\tilde{\rho}_t)
 \frac{\xi_{C_t}(h_1')}{\xi_t(h_2')}\bigg],
\end{align*}
from which the claim follows because $\tilde{\rho}_t\in[0,1]$ and
$h_1'\le h_2'$, so the latter fraction is upper bounded by one.
\end{proof}

\begin{lemma}
    \label{lem:rho-stuff} 
Suppose that Assumption~\ref{a:mixing-type} holds,
then, for any $t=1{:}T-1$,
\[
    \E[\rho_{t}] \ge \beta_N^{t} + \sum_{u=1}^{t-1}
    \big(\E[\tilde{\rho}_u] - 1 \big),
    \qquad\text{where}\qquad
    \beta_N \defeq \bigg( 1 + \frac{2c_*}{N-1}\bigg)^{-1}.
\]
\end{lemma} 
\begin{proof} 
We may apply Lemma~\ref{lem:ccpf-diff-bound-1} recursively
with
$h_1^{(u)} = h_2^{(u)} = \bar{G}_{u,t}$ where $\bar{G}_{u,t}(x_u)
\defeq \int Q_u(x_{u:u+1}) \bar{G}_{u+1,t}(x_{u+1})\ud x_{u+1}$ and
$\bar{G}_{t,t} = G_t$, leading to
\begin{align*}
    {h_2'}^{(u)}(x) &= G_{u-1}(x) \big[
    2 (N-1)^{-1} \sup_{x'} \bar{G}_{u,t}(x') + (M_u \bar{G}_{u,t})(x)\big]\\
    &\le G_{u-1}(x) (M_u \bar{G}_{u,t})(x) \bigg( 1 +
    \frac{2c_*}{N-1}\bigg)\\
    & = \bar{G}_{u-1,t} \beta_N^{-1},
\end{align*}
implying that
\[
    \E \bigg[
    \frac{\xi_{C_{u+1}}(\bar{G}_{u+1,t})}{\xi_{u+1}(\bar{G}_{u+1,t})}\bigg]
    \ge \E \bigg[
    \frac{\xi_{C_{u}}(\bar{G}_{u,t})}{\xi_{u}(\bar{G}_{u,t})}\bigg]
      \beta_N +
       (\E[\tilde{\rho}_u] - 1).\qedhere
\]
\end{proof} 

\begin{lemma}
    \label{lem:oneshot-with-rate-proof} 
Under Assumption~\ref{a:mixing-type},
\[
    \E[\rho_T] \ge 1 - 2^T (1- \beta_N^T)
    \ge 1 - \frac{2^T T}{(2c_*)^{-1}(N-1)+1}.
\]
\end{lemma} 
\begin{proof} 
The first inequality follows once we prove inductively that
$(1 - \E[\rho_t]) \vee (1 - \E[\tilde{\rho}_t])
\le 2^t (1 - \beta_N^t)$, which holds for $t=1$ by
Lemma~\ref{lem:rho-stuff} (which is symmetric wrt.~$\rho_t$ and
$\tilde{\rho}_t$). Then,
\begin{align*}
    1 - \E[\rho_t]
    &\le 1 - \beta_N^t + \sum_{u=1}^{t-1}
    \big(1-\E[\tilde{\rho}_t]\big) \\
    &\le (1 - \beta_N^t)\Big(1 + \sum_{u=1}^{t-1} 2^u \Big)
    \le 2^t (1- \beta_N^t),
\end{align*}
and the same bound applies to $1 - \E[\tilde{\rho}_t]$. The latter
bound follows as $1 - \beta_N^T \le T (1 - \beta_N)$.
\end{proof} 


\section{Unbiased estimator based on a coupled Markov kernel} 

We formalise here the construction of unbiased estimators of Markov
chain equilibrium expectations due to~\cite{glynn-rhee}, and
complement the results in~\cite{jacob-lindsten-schon}.
\begin{definition}[Coupling of probability measures] 
Suppose $\mu$ and $\nu$ are two probability measures on $\Sp$.
The set of couplings $\Gamma(\mu,\nu)$ consists of all probability
measures $\lambda$ on $\Sp\times\Sp$ with marginals
$\lambda(\uarg\times\Sp) = \mu$ and $\lambda(\Sp\times\uarg) = \nu$.
\end{definition} 

\begin{definition}[Coupled Markov kernel] 
Suppose $P$ is a Markov kernel on $\Sp$,  and
$\boldsymbol{P}$ is a Markov kernel on $\Sp\times\Sp$.
If $\boldsymbol{P}(x,\tilde{x}; \uarg) \in \Gamma\big(P(x,\uarg),
P(\tilde{x},\uarg)\big)$ for all $x,\tilde{x}\in\X$,
then, $\boldsymbol{P}$ is a \emph{coupled kernel} corresponding to $P$.
\end{definition} 

\begin{definition}[Coupling time] 
The \emph{coupling time} of the bivariate Markov chain
$(X_n,\tilde{X}_n)_{n\ge 0}$
is the random variable $\tau\defeq \inf \{n\ge 0\given X_k =
  \tilde{X}_{k}\text{ for all }k\ge n\}$.
\end{definition} 

\begin{theorem}
    \label{thm:coupling-unbiased} 
Let $P$ be an ergodic Markov kernel on $\X$ with invariant
distribution $\pi$ (that is, for all $x\in\X$,
$\| P^n(x,\uarg) - \pi \|_\tv\xrightarrow{n\to\infty} 0$), and suppose $\boldsymbol{P}$ is a corresponding
coupled Markov kernel. Let $\nu$ be any probability distribution on
$\X$, and suppose that $\lambda\in \Gamma(\nu P, \nu)$.

Consider a Markov chain $(X_n,\tilde{X}_n)_{n\ge 0}$ with initial
distribution $\lambda$ and transition probability $\boldsymbol{P}$,
and $h\in L^2(\pi)$.
Let the coupling time $\tau$
of $(X_n,\tilde{X}_n)_{n\ge 0}$
be a.s.~finite,
$\sup_{n\ge 0}\E[h^2(X_n)]<\infty$ for all $n\ge 0$ and
\begin{equation}
     \sup_{\{m,L \given L\ge m\}}
    \E[Z_{m,L}^2] < \infty,		\qquad\text{where}\qquad
    Z_{m,L} \defeq
    \sum_{n=m}^L [h(X_n) -
    h(\tilde{X}_{n})]\charfun{n< \tau}.
    \label{eq:cauchy}
\end{equation}
Then for all $b\ge 0$, $\E[Z_b]=\pi(h)$ and $\var(Z_b)<\infty$, where
\[
    Z_b \defeq h(X_b) + \sum_{n=b+1}^\infty [h(X_n) -
    h(\tilde{X}_{n})]\charfun{n< \tau}.
\]
\end{theorem} 
\begin{proof} 
Note that $X_n \overset{d}{=} \tilde{X}_{n+1}$ for all $n\ge 0$,
and so for any $m>b$,
\begin{equation*}
    \E[h(X_m)] = \E\bigg[\overbrace{h(X_b) + \sum_{n=b+1}^m [h(X_n) -
    h(\tilde{X}_{n})]\charfun{n< \tau}}^{\eqdef \zeta_m}\bigg].
\end{equation*}
Fix $b>0$. By~\eqref{eq:cauchy}, $\{ \zeta_m \}$ is a bounded  sequence in $L^2$. The almost sure finiteness of the stopping time ensures $Z_b$ is well defined and $\zeta_m \rightarrow Z_b$ almost surely. Thus $Z_b$ is also square integrable and $\zeta_m \rightarrow Z_b$ in $L^1$.
Since $\nu P^{n}$ converges to $\pi$ in total variation, the assumptions $\sup_n \nu P^{n}(h^2)<\infty$ and $\pi(h^2)<\infty$ imply
$\nu (P^{n} h) \rightarrow \pi(h)$. Finally,
$\E[Z_b]=\pi(h)$ follows from $\E[\zeta_m] = \E[h(X_m)] = \nu (P^{m+1} h) \rightarrow \pi(h)$.
\end{proof} 

\begin{lemma}
    \label{lem:variance-convergence} 
Letting $\bar{h} \defeq h - \pi(h)$, the variance of the estimator satisfies
\[
    |\var(Z_b) - \var_\pi\big(h(X)\big)|
    \le  \|\bar{h}\|_\infty^2
    \| \nu P^{b+1} - \pi \|_\mathrm{tv}
    + 2 \| \bar{h}\|_\infty (\E Z_{b+1,\infty}^2)^{1/2}
    + \E Z_{b+1,\infty}^2.
\]
\end{lemma} 
\begin{proof} 
We may write
\[
    \var(Z_b)
    = \E\big( \bar{h}(X_b) + Z_{b+1,\infty} \big)^2
    = \E \bar{h}^2(X_b) + 2 \E[\bar{h}(X_b)Z_{b+1,\infty}]
    + \E Z_{b+1,\infty}^2.
\]
Let $X \sim \pi$, then
\[
    |\var(Z_b) - \var(h(X))|
    \le | \E \bar{h}^2(X_b) - \E \bar{h}^2(X) | + 2 \|\bar{h}\|_\infty
    \E|Z_{b+1,\infty}| + \E Z_{b+1,\infty}^2.
\]
The first term is upper bounded by $\|\bar{h}\|_\infty^2
\| \nu P^{b+1}  - \pi\|_\mathrm{tv}$, and $(\E|Z_{b+1,\infty}|)^2\le
  \E Z_{b+1,\infty}^2$.
\end{proof} 

Below, we use $\| h \|_{\mathrm{osc}} \defeq
\sup_{x,y\in\X} | h(x) - h(y)|$, and
$\|h\|_\infty = \sup_{x\in\X}
|h(x)| \ge  \| h \|_\mathrm{osc} /2 $.

Under the following assumed distribution on the coupling time, not only is the sequence $Z_{m,L}$ defined in~\eqref{eq:cauchy} uniformly square integrable, the corresponding sequence $\{ \zeta_m \}$  is an $L^2$ Cauchy sequence.
\begin{lemma}
    \label{lem:cauchy-terms} 
Suppose that there exist $C<\infty$ and $\lambda\in[0,1)$ such that
for all $n\in\N$, $\P(\tau>n)\le C \lambda^{n}$, then
$\E[Z_{m,L}^2] \le  2 C\lambda^m(1-\lambda)^{-2} \| h
  \|_{\mathrm{osc}}^2$ for all $L\ge m\ge 1$.
\end{lemma} 
\begin{proof} 
Let $\Delta h_n\defeq h(X_n)-h(\tilde{X}_n)$, then $|\Delta h_n|\le \|
h \|_{\mathrm{osc}}$, and so
\begin{align}
    \E[Z_{m,L}^2] &= \E\bigg[\sum_{n,\ell=m}^L \Delta h_n \Delta h_\ell
    \charfun{\tau >n\vee
           \ell}\bigg]
   \le \| h \|_{\mathrm{osc}}^2 \sum_{n,\ell=m}^L \P( \tau > n\vee
   \ell).
   \label{eq:holderisable}
\end{align}
The latter sum may be upper bounded by
\begin{align*}
C \sum_{n,\ell=m}^L \lambda^{n\vee \ell}
   \le C \sum_{i=0}^\infty (2i+1)\lambda^{i+m}
   = C \lambda^{m} \bigg(2 \frac{\lambda}{(1-\lambda)^2} +
   \frac{1}{(1-\lambda)}\bigg).
\end{align*}
Simple calculation yields the desired bound.
\end{proof} 

\begin{lemma}
    \label{lem:marginal-corollary} 
Suppose $\mathbf{P}$ is a coupled kernel corresponding to a
$\pi$-ergodic Markov kernel $P$. Let $\tau_{x,\tilde{x}}$
stand for the coupling time of the Markov chain
$(X_n,\tilde{X}_n)_{n\ge 0}$ with transition
probability $\mathbf{P}$ and with $(X_0,\tilde{X}_0)\equiv
(x,\tilde{x})$.
Then,
\[
    \| P(x,\uarg) - \pi \|_\mathrm{tv}
    \le 2 \sup_{\tilde{x}\in\X} \P(\tau_{x,\tilde{x}} > n).
\]
\end{lemma} 
\begin{proof} 
Let $\tau_x$ stand for the coupling time of
$(X_n,\tilde{X}_n)_{n\ge 0}$ with $X_0\equiv x$ and $\tilde{X}_0\sim
\pi$. By the standard coupling inequality,
$\| P(x,\uarg) - \pi\|_\mathrm{tv} \le 2 \P(\tau_x > n)$, and
\[
    \P(\tau_x > n)
    = \int \P(\tau_x > n\mid \tilde{X}_0=\tilde{x}) \pi(\ud \tilde{x})
    = \int \P(\tau_{x,\tilde{x}} > n) \pi(\ud \tilde{x}).\qedhere
\]
\end{proof} 

\begin{lemma}
\label{lem:var_bound} 
Let $\tau_{x,\tilde{x}}$ be as in Lemma~\ref{lem:marginal-corollary},
and assume that there exist $C\in[1,\infty)$ and $\lambda\in(0,1)$ such that
$\sup_{\tilde{x},x\in\X} \P(\tau_{x,\tilde{x}} > n) \le C \lambda^n$.
Then,
\[
    |\var(Z_b) - \var_\pi(h(X)) |
    \le
     \frac{16C}{(1-\lambda)^2}
    \lambda^{(b+1)/2}
    \| \bar{h} \|_\infty^2.
\]
\end{lemma} 
\begin{proof} 
Using Lemma~\ref{lem:variance-convergence} together with
Lemma~\ref{lem:marginal-corollary} and Lemma~\ref{lem:cauchy-terms}
yields
\begin{align*}
    |\var(Z_b) - \var_\pi(h(X)) |
    &\le
    2 C {\| \bar{h}\|}^2_\infty \lambda^{b+1} + 2 \|\bar{h}\|_\infty
    \bigg(\frac{2 C\lambda^{b+1}}{(1-\lambda)^2}
    \| h \|_{\mathrm{osc}}^2\bigg)^{\frac{1}{2}} \\
    &\phantom{\le}
    + \frac{2 C\lambda^{b+1}}{(1-\lambda)^2}
    \| h \|_{\mathrm{osc}}^2.
\end{align*}
The claim follows easily, because
$\| h \|_\mathrm{osc} \le
2 \| \bar{h}\|_\infty$.
\end{proof} 


\section{Supplementary simulation results}
\label{sec:supp_sims} 

Figure~\ref{fig:mean-coupling-times-ind} corresponds to
Figure~\ref{fig:mean-coupling-times} but where common random numbers are not
used in Line~\ref{line:uncoupled-mutation} of Algorithm~\ref{alg:ccxpf}.
We observe that AT and AS show much worse performance, whereas BS is hardly
affected.

\begin{figure}
    \centering
    \begin{subfigure}[b]{0.49\linewidth}
    \centering
    \includegraphics[width=\linewidth]{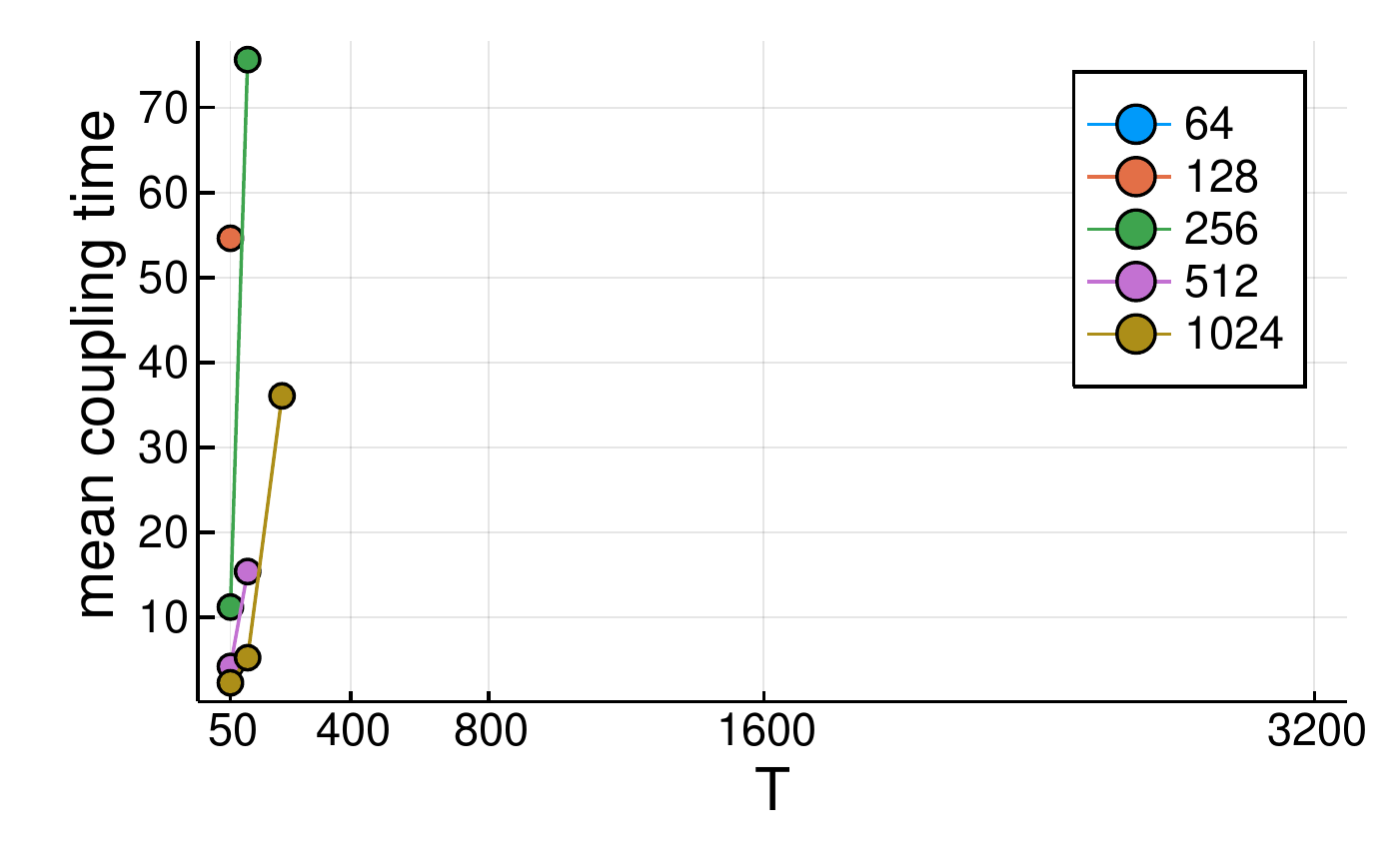}
    \caption{Linear Gaussian: AT}
    \end{subfigure}
    \begin{subfigure}[b]{0.49\linewidth}
    \centering
    \includegraphics[width=\linewidth]{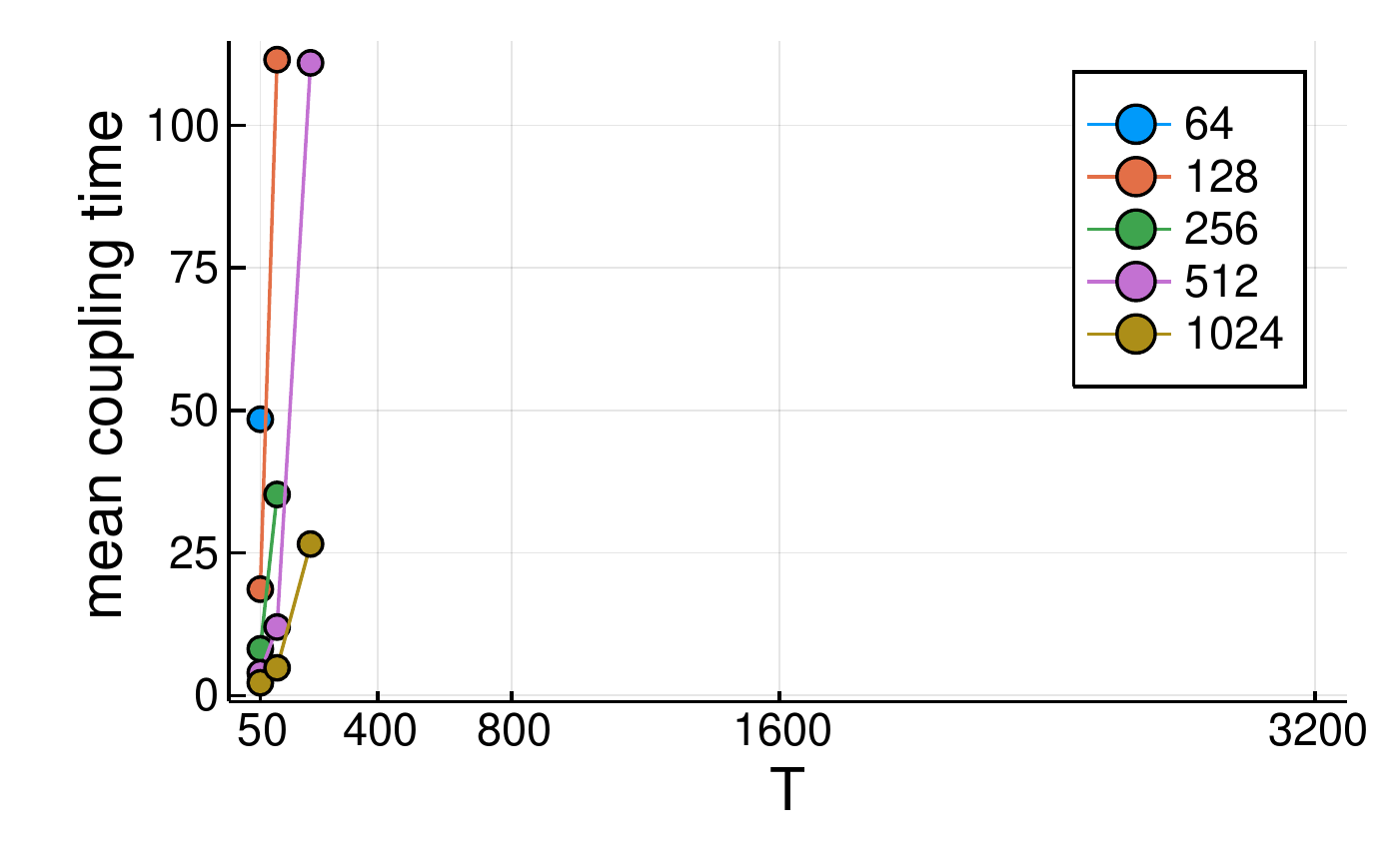}
    \caption{Linear Gaussian: AS}
    \end{subfigure}
    \begin{subfigure}[b]{0.49\linewidth}
    \centering
    \includegraphics[width=\linewidth]{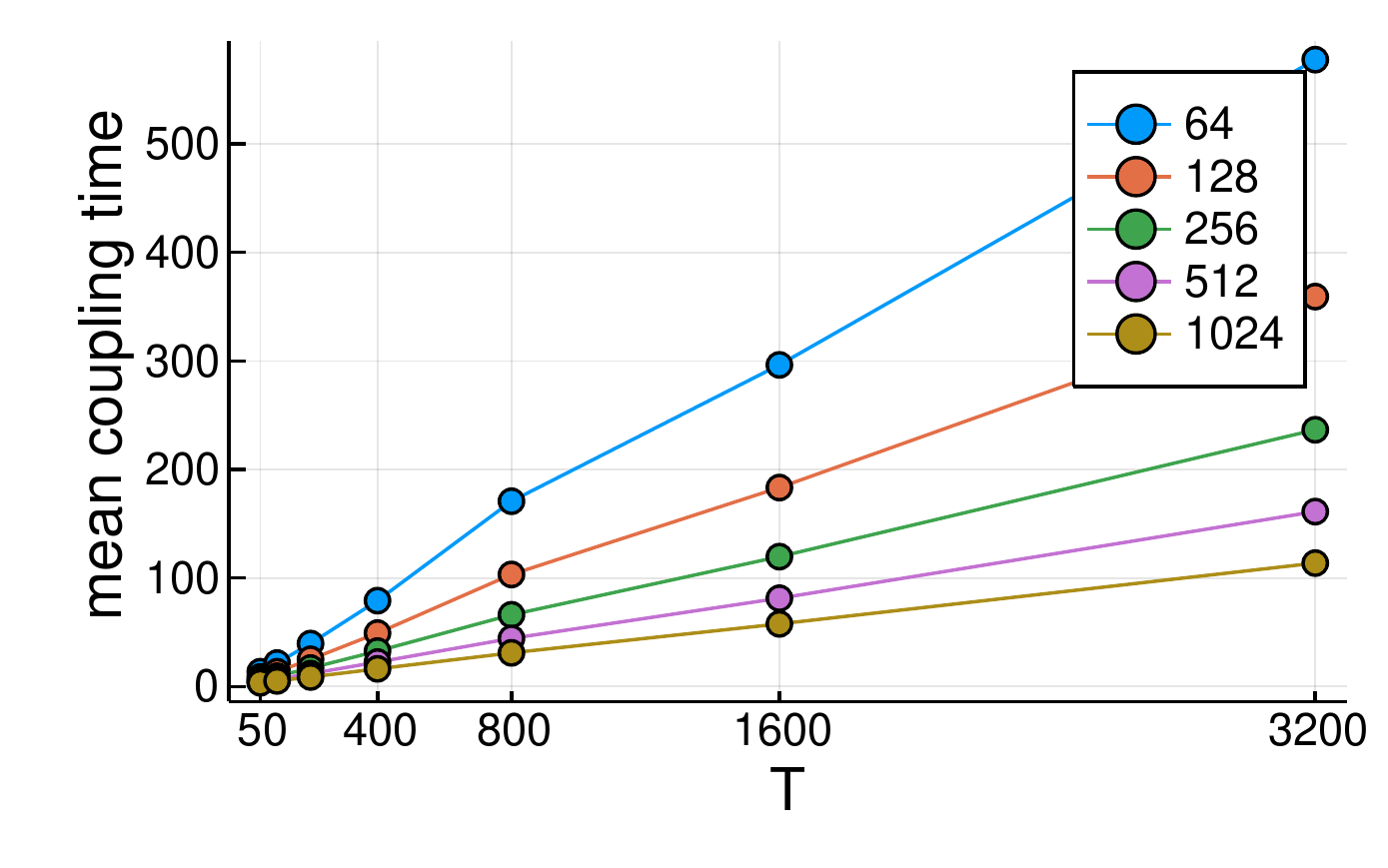}
    \caption{Linear Gaussian: BS}
    \end{subfigure}
    \begin{subfigure}[b]{0.49\linewidth}
    \centering
    \includegraphics[width=\linewidth]{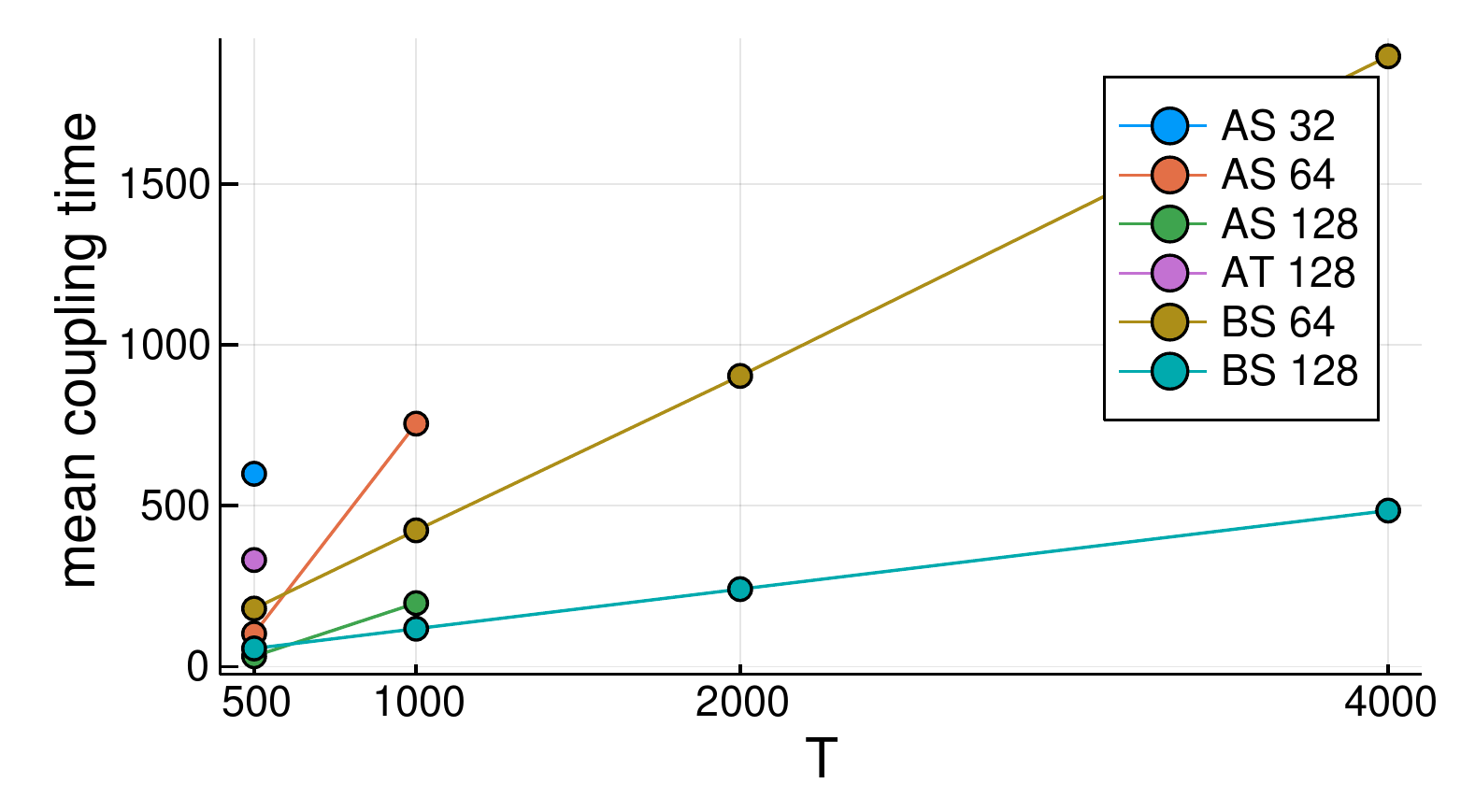}
    \caption{Simple homogeneous model}
    \end{subfigure}
    \caption{Mean coupling times associated with ancestor tracing (AT), ancestor
      sampling (AS) and backward sampling (BS). For (d), the lines are coloured
      according to the type of algorithm and the number of particles $N$.}
    \label{fig:mean-coupling-times-ind}
\end{figure}


\end{document}